\documentclass{article}

\usepackage[a4paper, total={6in, 8in}]{geometry}

\usepackage[square,sort,comma,numbers]{natbib}

\usepackage{latexsym}

\usepackage{bm}

\usepackage{doi}

\usepackage{epsfig}

\usepackage{enumerate}
\usepackage[font=scriptsize,labelfont=bf]{caption}
\usepackage{pst-all}
\usepackage{multido}
\usepackage{pst-blur}
\usepackage{color}
\usepackage{hyperref}

\usepackage[normalem]{ulem}

\usepackage{color}
\usepackage{graphicx}
\usepackage{amsmath}
\usepackage{amssymb}
\usepackage{amsthm}

\usepackage[ruled,vlined]{algorithm2e}
\SetKwProg{Init}{Initialization:}{}{}

\usepackage{amsfonts,amsgen,amsopn}
\usepackage{epsfig}

\usepackage{url}

\usepackage{graphicx}

\usepackage{latexsym}
\usepackage{amsfonts,amsgen,amsopn}
\usepackage{epsfig}
\usepackage{color}

\usepackage{subfig}

\usepackage{authblk}

\usepackage{csquotes}
\usepackage{verbatim}

\usepackage{comment}

\usepackage[square,sort,comma,numbers]{natbib}

\usepackage{arydshln}

\newtheorem{theorem}{Theorem}[section]
\newtheorem{lemma}[theorem]{Lemma}

\newtheorem{corollary}[theorem]{Corollary}
\newtheorem{proposition}[theorem]{Proposition}

\newtheorem{remark}{Remark}
\newtheorem{conjecture}{Conjecture}

\newtheorem{definition}[theorem]{Definition}

\newtheorem{example}[theorem]{Example}

\newcommand{\T}{{\mathcal T}}
\newcommand{\C}{{\mathcal C}}
\newcommand{\NN}{{\mathbb N}}
\newcommand{\sS}{{\mathcal S}}

\newtheorem{problem}{Problem}

\title{Perfect taxon sampling and fixing taxon traceability: Introducing a class of phylogenetically decisive collections of taxon sets}

\author[$\ast$,1]{Mareike Fischer}
\author[2,3]{Janne Pott}

\affil[1]{Institute of Mathematics and Computer Science, University of Greifswald, Greifswald, Germany, mareike.fischer@uni-greifswald.de and email@mareikefischer.de}
\affil[2]{Institute for Medical Informatics, Statistics and Epidemiology, Leipzig University, Leipzig, Germany}
\affil[3]{MRC Biostatistics Unit, University of Cambridge, Cambridge, UK}

\date{}

\begin{document}
\maketitle

\begin{abstract}
Phylogenetically decisive collections of taxon sets have the property that if trees are chosen for each of their elements, as long as these trees are compatible, the resulting supertree is unique. This means that as long as the trees describing the phylogenetic relationships of the (input) species sets are compatible, they can only be combined into a common supertree in precisely one way. This setting is sometimes also referred to as \enquote{perfect taxon sampling}. While for rooted trees, the decision if a given set of input taxon sets is phylogenetically decisive can be made in polynomial time, the decision problem to determine whether a collection of taxon sets is phylogenetically decisive concerning \emph{unrooted} trees is unfortunately coNP-complete and therefore in practice hard to solve for large instances. This shows that recognizing such sets is often difficult. In this paper, we explain phylogenetic decisiveness and introduce a class of input taxon sets, namely so-called \emph{fixing taxon traceable} sets, which are guaranteed to be phylogenetically decisive and which can be recognized in polynomial time. Using both combinatorial approaches as well as simulations, we compare properties of fixing taxon traceability and phylogenetic decisiveness, e.g., by deriving lower and upper bounds for the number of quadruple sets (i.e., sets of 4-tuples) needed in the input set for each of these properties. In particular, we correct an erroneous lower bound concerning phylogenetic decisiveness from the literature. 

We have implemented the algorithm to determine if a given collection of taxon sets is fixing taxon traceable in \textsf{R} and made our software package \verb+FixingTaxonTraceR+ publicly available. 
\end{abstract}

\section{Introduction}\label{introduction}

Reconstructing the so-called \enquote{Tree of Life}, i.e., the phylogenetic tree displaying all living species on earth, is one of the main challenges of biological sciences to-date. Sequence data (e.g., DNA or protein data) on some clusters of species are already available
in databases like GenBank \cite{genbank}, UniProt \cite{uniprot} and KEGG \cite{kegg}, and there are algorithms available to reconstruct the tree of each cluster. In many studies, data
from different loci are combined by building trees from each locus
and subsequently combining them into a so-called \enquote{supertree}. In this setting, it is common that the supertree contains \emph{all} taxa whereas the input trees for each individual locus often do not contain all taxa under consideration. While it is even possible that these input trees are incompatible with one another (which means there exists no perfect supertree, i.e., a supertree containing all the input trees as subtrees), even in the case of compatibility it is not always clear which supertree is best (as there may be more than one tree on the entire taxon set containing all input trees).

For example, suppose we have sampled two sets of species, say $\sS=\{\{1,2,3,4\},$  $\{1,2,3,5\}\}$. So the set of taxa we are investigating is $X=\{1,2,3,4,5\}$, and we want to know if the relationship of these five species is uniquely determined by the two sample sets we have. Then, if the input trees for the two sets are the trees depicted in Figure \ref{fig_NonDec1}, the supertree is \emph{not} unique: Actually, three trees are possible as depicted in Figure \ref{fig_NonDec2}, as all of them are compatible with the input trees. 

\begin{figure}[ht]      \centering\vspace{0.5cm} 
   \includegraphics[width=6cm]{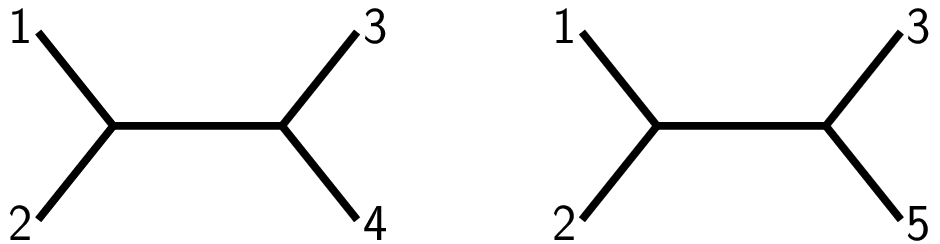}  
    \caption{Two possible input trees for $\sS = \{\{1,2,3,4\},\{1,2,3,5\}\}$ and $X=\{1,2,3,4,5\}$. These trees are compatible, i.e., there exists a supertree on taxon set $X$ containing both of them as subtrees, but this tree is \emph{not} unique, cf. Figure \ref{fig_NonDec2}.} \label{fig_NonDec1}
  \end{figure}

\begin{figure}[ht]      \centering\vspace{0.5cm} 
   \includegraphics[width=8cm]{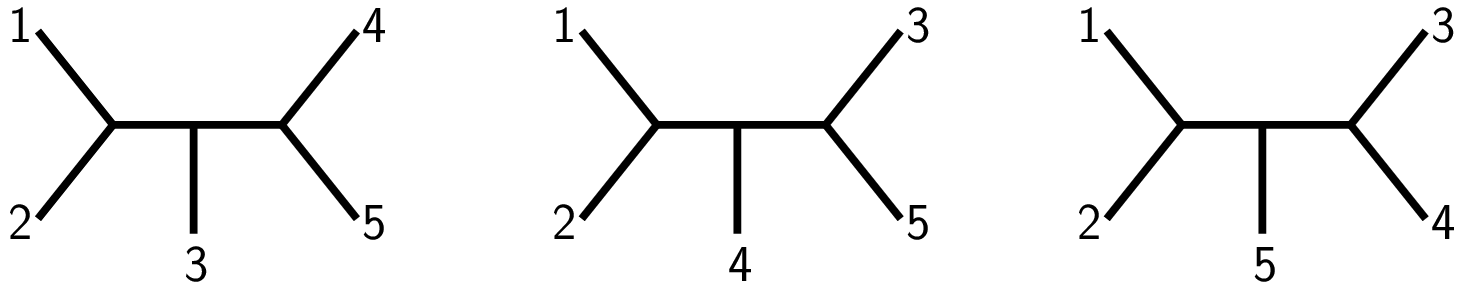} 
    \caption{Three trees with the properties that in each of them, deleting the edge leading to leaf 5 and suppressing the resulting degree-2 vertex will yield the left tree from Figure \ref{fig_NonDec1}, whereas doing the same with leaf 4 will yield the second tree from said figure. This shows that all three trees are supertrees displaying both trees from Figure \ref{fig_NonDec1}, so these trees have more than one possible supertree.  }\label{fig_NonDec2}
  \end{figure}

However, if the input trees instead are the trees depicted on the left in Figure \ref{fig_NonDec3}, then the only possible supertree is the tree on the right-hand side of the same figure. This means that the collection of taxon sets $\sS$ does \emph{not} itself carry the information if the supertree is unique, because it depends on the particular choice of input trees.

\begin{figure}[ht]      \centering\vspace{0.5cm} 
    \includegraphics[width=8cm]{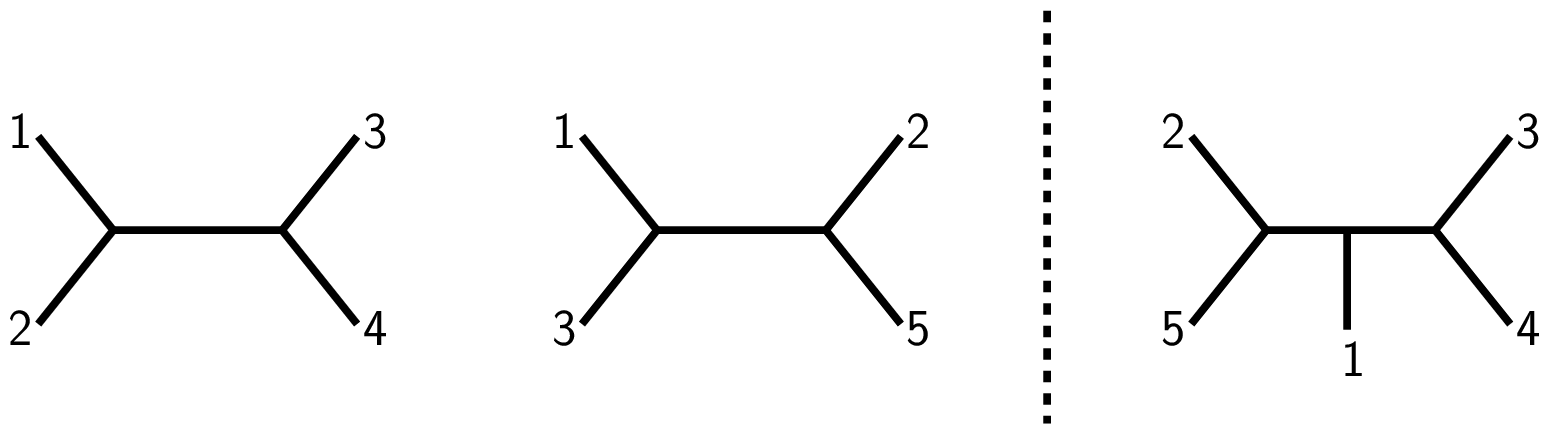} 
    \caption{Two input trees (left) for $\sS = \{\{1,2,3,4\},\{1,2,3,5\}\}$ and $X=\{1,2,3,4,5\}$. Note that the second tree differs from the second tree in Figure \ref{fig_NonDec1} as  leaves 2 and 3 are swapped. The unique supertree of these two input trees is depicted on the right. Note that the uniqueness of the supertree can be easily verified by attaching leaf 5 to all edges of the first input tree and checking if the subsequent deletion of leaf 4 leads to the second input tree or not. It turns out that the only way to combine these two trees is to attach leaf 5 to the edge incident with leaf 2 in the first tree.} 
    \label{fig_NonDec3}
  \end{figure}

In 2010, Sanderson and Steel \cite{sanderson_steel_2010} mathematically characterized so-called \emph{phylogenetically decisive} collections of taxon sets. These sets consist of input taxon sets which have the property that \emph{all} possible compatible trees chosen for the input sets lead to a unique supertree. As an example, consider again  taxon sets $\{1,2,3,4\}$ and $\{1,2,3,5\}$. If we now additionally consider the sets $\{1,3,4,5\}$ and $\{2,3,4,5\}$, i.e., $\sS'=\{\{1,2,3,4\},\{1,2,3,5\},\{1,3,4,5\},\{2,3,4,5\}\}$, it can be shown that \emph{regardless of the particular choice of input trees} the supertree will always be unique (as long as the chosen input trees are compatible). While this manuscript will provide some insight into how to verify this in a more efficient manner, the naive approach of exhaustively checking all possible choices of input trees is provided by Table \ref{tab:examplealltrees} in the appendix. Note that there are only $\binom{5}{4}=5$ different quadruples on five taxa, four of which are contained in $\sS'$. We will later see that these four quadruples carry sufficient information to resolve the missing one. However, the fact that $\sS'$ is phylogenetically decisive can also be verified by explicitly looking at all possible choices of input trees. For example, if we consider the input trees from Figure \ref{fig_NonDec1} again together with the two trees depicted by Figure \ref{fig_NonDec5}, the only possible supertree is the leftmost tree in Figure \ref{fig_NonDec2}. The fact that this works for all possible choices of compatible input trees for the taxon sets in $\sS'$ means that $\sS'$ is phylogenetically decisive. Note that intuitively, this has to do with certain overlap properties of the elements of $\sS'$, on which we elaborate further in the present manuscript.

\begin{figure}[ht]      \centering\vspace{0.5cm} 
   \includegraphics[width=6cm]{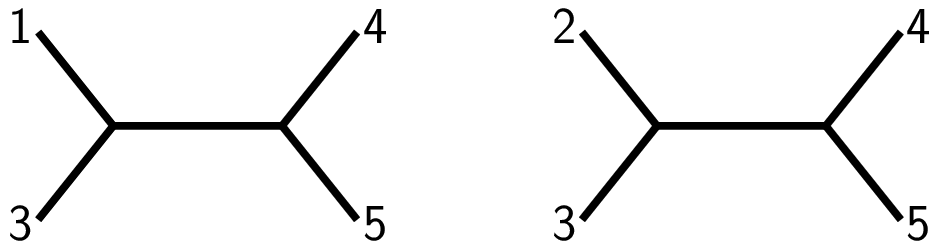} 
    \caption{Two trees on taxon sets $\{1,3,4,5\}$ and $\{2,3,4,5\}$, respectively, which together with the trees from Figure \ref{fig_NonDec1} lead to the unique supertree depicted on the very left of Figure \ref{fig_NonDec2}.} \label{fig_NonDec5}
  \end{figure}

The example of $\sS'$ illustrates that the decision which taxa to sample for each input set $X$ may already ensure that the supertree of all input trees is unique -- even before a single tree is analyzed! In this context, a collection of taxon sets which is phylogenetically decisive can also be referred to as a set of perfect taxon samples. 

However, the computational complexity of the decision problem to decide whether a given collection of taxon sets is phylogenetically decisive when unroooted trees are considered remained unknown until recently, even though it is of high biological interest and even though the problem can be easily shown to be solvable in polynomial time whenever rooted trees are considered.  The problem can be specified as follows:

\begin{problem}[(Unrooted) phylogenetic decisiveness decision problem] \label{decprob} Given a taxon set $X$ with $\vert X\vert =n$ and a set $\sS=\{Y_1,\ldots,Y_k\}$ of subsets of $X$, decide if $\sS$ is phylogenetically decisive considering unrooted trees. 
\end{problem}

In their 2010 study, Steel and Sanderson showed that a collection of taxon sets is phylogenetically decisive (concerning unrooted trees) if and only if it satisfies the so-called \enquote{four-way partition property}. However, this characterization requires exponential computation time, because checking the four-way partition property requires  checking all possible partitions of $X$ into four non-empty subsets. So this 2010 study left the question of the computational complexity of Problem \ref{decprob} open. However, in 2022, it was shown by Parvini, Braught and Fern\'andez-Baca \cite[Theorem 3]{parvini2022} that the problem is coNP-complete. In order to show their result, they exploited the equivalence of the negation of Problem \ref{decprob} to the so-called \emph{No rainbow coloring problem}:

\begin{problem}[No rainbow coloring problem] \label{rainbowprob} Given a 4-uniform hypergraph, i.e., a pair $(V,E)$ of vertices $V$ and a set of hyperedges $E$ where the elements of $E$ consist of four elements of $V$ each, with vertices $V$, decide if there is a partition of $V$ into four non-empty subsets $A,\ B, \ C, \ D \subseteq V$ such that no edge is \enquote{rainbow-colored}, i.e., no edge contains one vertex from each of $A$, $B$, $C$ and $D$.
\end{problem}

While the study by Parvini et al. \cite{parvini2022} finally answers the complexity question for Problem \ref{decprob}, the answer is not ideal for biological applications. In order to sample in a \enquote{perfect} way to guarantee unique supertrees, it is important to be able to efficiently identify phylogenetically decisive collections of taxon sets. Luckily, an integer linear program is available and Problem \ref{decprob} is actually fixed parameter tractable in the number of loci (both results can be found in \cite{parvini2022}), which does give some hope for practical applications, but as the problem is coNP-complete, it cannot be solved efficiently and will thus cause problems for large data sets.

It is therefore the aim of the present manuscript to introduce the concepts of so-called \emph{fixing taxa} (Section \ref{sec:CQ}) as well as \emph{fixing taxon traceable collections of taxon sets} (Section \ref{sec:FTTintro}). The latter will be shown to be collections of taxon sets that are guaranteed to be phylogenetically decisive (Theorem \ref{thm_FTTmain}) and which can be identified in polynomial time (Section \ref{sec_alg}). We also analyze this interesting class of collections of taxon sets more in-depth and compare it to the class of phylogenetically decisive collections of taxon sets. Moreover, we derive lower and upper bounds for the input information needed to exclude or guarantee fixing taxon traceability or phylogenetic decisiveness, respectively Sections \ref{sec_FTTbounds} and \ref{sec_PDbounds}. In particular, with a simple counterexample, we show that a lower bound on the number of required quadruples for phylogenetic decisiveness known from the literature  \cite[Theorem 4(i)]{parvini2022} is unfortunately erroneous, and we subsequently present a corrected (but not tight) lower bound (Theorem \ref{thm_lowerboundphylodec}).

We implemented our algorithm to detect fixing taxon traceable collections of taxon sets as well as an adaptation of Algorithm 2 from \cite{parvini2022} \verb+Find4NRC+. The latter gives a characterization of phylogenetically decisive collections of taxon sets. Our resulting  software package \verb+FixingTaxonTraceR+, which was implemented in \textsf{R} \cite{RCoreTeam2019}, has been made publicly available \cite{githubsoftware}. Last but not least, we used \verb+FixingTaxonTraceR+ to perform a simulation study (cf. Section \ref{sec_diff}) to compare the sizes of the class of phylogenetically decisive collections of taxon sets and its subclass of fixing taxon traceable collections of taxon sets \cite{githubsimulation}. 

\section{Preliminaries}
\label{sec_prelim}

\subsection{Definitions and concepts}
\label{sec_def}
We first introduce some notations and definitions required for presenting our results. We start with the standard notion of phylogenetic trees.

A \emph{binary phylogenetic $X$-tree $\T$}, also often referred to as \emph{binary phylogenetic tree on $X$}, is a connected acyclic graph in which all internal nodes are of degree 3 and which there are $\vert X\vert =n\geq 1$ nodes of degree at most 1, namely the so-called \emph{leaves} or \emph {taxa}, which are bijectively labelled with the $n$ elements of a set $X$, which is also often referred to as \emph{taxon set}. Without loss of generality, we assume $X=\{1,\ldots, n\}$. Whenever there is no ambiguity, we refer to a binary  phylogenetic $X$-tree as \emph{tree} for short. 

Now, let $X=\{1,\ldots, n\}$, let $Y \subset X$ be a subset of $X$ and let $\T$ be a binary phylogenetic $X$-tree. Then, we denote by $\T\vert _{Y}$ the tree which can be derived from $\T$ by deleting all elements of $X\setminus Y$ and by suppressing all resulting nodes of degrees 2 and by deleting all potentially resulting degree-1 vertices which are not labelled by $X$. Note that this procedure is equivalent to taking the minimal spanning tree of $Y$ in $\T$ and suppressing its degree-2 vertices. We also say that $\T$ \emph{displays} $\T\vert _Y$. For example, the phylogenetic $X$-tree (with $X=\{1,2,3,4,5\}$) depicted on the left in Figure \ref{fig_NonDec2} displays all four trees depicted in Figures \ref{fig_NonDec1} and \ref{fig_NonDec5}. 

Let $Y_1,\ldots,Y_k \subset X$, for some $k \in \NN$ and let $\T_1,\ldots,\T_k$ be binary phylogenetic trees on $Y_1,\ldots,Y_k$, respectively. If there exists a binary phylogenetic tree $\T$ which displays all trees $\T_1,\ldots,\T_k$, we call $\T$ a \emph{supertree} of $\T_1,\ldots,\T_k$. Moreover, a set of trees $\{\T_1,\ldots,\T_k\}$ on taxon sets $Y_1, \ldots, Y_k$, respectively, is called \emph{compatible} if there exists a supertree $\T$ on taxon set $X:=\bigcup_{i=1}^k Y_i$ displaying all of the trees 
$\T_1,\ldots,\T_k$.

We are now in a position to define phylogenetic decisiveness.

\begin{definition}\label{def:phylodec} A collection $\sS=\{Y_1,\ldots,Y_k\}$ of subsets of $X$ (i.e., $Y_i \subseteq X \hspace{0.2cm}\forall i=1,\ldots,k$) is said to be \emph{phylogenetically decisive} if for every binary phylogenetic $X$-tree $\T$, the collection $\{\T\vert _{Y_i}: Y_i \in \sS\}$ characterizes $\T$ (i.e., $\T$ is the only tree on $X$ displaying these trees).\end{definition}

Last but not least, in order to work with phylogenetic decisiveness, we require the following crucial concept,  which has similarly been defined in the context of phylogenetic groves, e.g., in \cite{ane_eulenstein} and  \cite{fischer_groves}.

\begin{definition}[Cross quadruples and cross $c$-tuples]\label{CT} \label{CQ} Let $k \in \mathbb{N}_{\geq 1}$. Let $X$ be a (taxon) set and let $\sS=\{Y_1,\ldots,Y_k\}$, where $Y_i \subseteq X \hspace{0.2cm}\forall \ i=1,\ldots,k$. Then, a set $C \subseteq X$ with $\lvert C\rvert=c$ such that $C \not\subseteq Y_i$ for all $i=1,\ldots,k$ is called a \emph{cross $c$-tuple} of $\sS$ or simply \emph{cross tuple} of $\sS$. If $c=4$, $C$ is also often referred to as \emph{cross quadruple} or \emph{CQ} for short.\end{definition}

 Note that the elements of $\mathcal{S}$ are allowed to be empty -- in this case, the empty tree carrying no information on any $c$-tuple is the only choice for such elements. Generally, $\mathcal{S}$ can be thought of as the set of input sets, which all carry some information on taxa under consideration, whereas cross $c$-tuples can be thought of as the $c$-tuples for which no information is present in any of the input sets. Thus, if such tuples are to be resolved uniquely by all possible supertrees, there must be additional information on these tuples stemming from the fact that they are being analyzed together. As we will elaborate, the intersections of the elements of $\mathcal{S}$ play an important role in this regard.

\subsection{Known results on phylogenetic decisiveness}
\label{sec_known}

In 2010, a characterization of phylogenetic decisiveness, which we will use throughout this manuscript, was given by \citet{sanderson_steel_2010} in terms of the following theorem.

\begin{theorem}[adapted from Theorem 2 in \citet{sanderson_steel_2010}, \enquote{Four-way partition property}]$\mbox{}$\\ \label{sanderson_steel}  A collection $\sS=\{Y_1,\ldots,Y_k\}$ of subsets of a taxon set $X$ is phylogenetically decisive if and only if it satisfies the four-way partition property, i.e., if for all partitions of $X$ into four non-empty and non-overlapping subsets $X_1$, $X_2$, $X_3$ and $X_4$, denoted $\pi=X_1\vert X_2\vert X_3\vert X_4$,  there exist taxa $x_i \in X_i$ (for $i =1,2,3,4$) such that the quadruple $\{x_1,x_2,x_3,x_4\}$ is a subset of $Y_j $ for some $j \in \{1,\ldots,k\}$. For such a quadruple, we also say that it \emph{covers} partition $\pi$. 
\end{theorem}

Another result we will use later on is the following, which can be found in \cite{moan}. 

\begin{theorem}[adapted from Theorem 3.10 in \citet{moan}] \label{thm_moan} For every $n\geq 4$, the number $f_n = \binom{n}{4}-(n-4)$ is the smallest number such that
every collection of quadruples $\sS$ on a taxon set $X=\{1,\ldots,n\}$ with $\vert \sS \vert \geq f_n$ is phylogenetically decisive.
\end{theorem}

Before we continue, we present an example to illustrate Theorem \ref{thm_moan}, as this theorem is a crucial ingredient for our own results later on.

\begin{example} \label{CQpossible} Let $X=\{1,2,3,4,5\}$. Then, the set 
$\sS'=\{\{1,2,3,4\},\{1,2,3,5\},$ $\{1,3,4,5\},\{2,3,4,5\}\}$ known from Section \ref{introduction} is phylogenetically decisive even though it has a cross quadruple, namely $\{1,2,4,5\}$. This can be directly verified with Theorem \ref{thm_moan}, which states that if we have at least $f_5=\binom{5}{4}-(5-4)=5-1=4$ quadruples on taxon set $\{1,2,3,4,5\}$  (which is the case here as $\sS'$ contains four quadruples) this set of quadruples is phylogenetically decisive. 
\end{example}

Last but not least, recall that, as mentioned above, it was shown by \citet{parvini2022} that Problem \ref{decprob} is coNP-complete. This is why we are aiming at presenting a  subclass of phylogenetically decisive collections of taxon sets which  can be identified in polynomial time.

\section{Results}\label{results}

 It is the main aim of this section to introduce \emph{fixing taxon traceable} collections of taxon sets -- a class of collections of taxon sets for which phylogenetic decisiveness can be guaranteed and which can be identified in polynomial time.  However, before we can formally define fixing taxon traceability in Section \ref{sec:FTTintro} and present an  algorithm to detect fixing taxon traceability in Section \ref{sec_alg}, we first introduce the concept of fixing taxa in Section \ref{sec:CQ}. Later on, we identify several bounds both for fixing taxon traceability in Section \ref{sec_FTTbounds} and  phylogentic decisiveness in Section \ref{sec_PDbounds}, before we analyze the differences of these two concepts in Section \ref{sec_diff} both analytically as well as statistically.

\subsection{Cross quadruples and fixing taxa}\label{sec:CQ}

In this section, we consider cross quadruples and subsequently define fixing taxa, which -- as we will show -- provide a way to resolve cross quadruples. As a first step, we provide a simple sufficient criterion for phylogenetic decisiveness, which also links the rooted and the unrooted cases of the underlying decision problem.\footnote{Note that it can easily be seen that in the rooted setting, a collection of taxon sets is phylogenetically decisive if and only if it contains all triples, i.e., if there are no \enquote{cross triples}. This is the reason why the decisiveness question can be answered in polynomial time in the rooted case \cite[Theorem 2(ii)]{parvini2022}.} 

\begin{proposition}\label{phylodecisive2}
Let $\sS=\{Y_1,\ldots,Y_k\}$ be a set of subsets of $X$ such that there exists no cross quadruple of $\sS$. Then, $\sS$ is phylogenetically decisive. \end{proposition}

\begin{proof} The statement is a direct consequence of Theorem \ref{sanderson_steel}: If $\sS$ has no cross quadruple, this implies that all possible quadruples are present in $\sS$. Given any 4-partition $X_1\vert X_2 \vert X_3 \vert X_4$ of $X$ and any four taxa $x_i \in X_i$ for $i=1,\ldots,4$, this implies that the 4-partition is covered by the quadruple $\{x_1,x_2,x_3,x_4\}$, which is present in $\sS$, in the sense of Theorem \ref{sanderson_steel}. This completes the proof.
\end{proof}

The fact that phylogenetic decisiveness is coNP-complete in the unrooted case immediately implies that the criterion of having no cross quadruples, which is sufficient as stated in Proposition \ref{phylodecisive2}, cannot be generally necessary: As can be seen in Example \ref{CQpossible},  cross quadruples do not necessarily destroy phylogenetic decisiveness. However, this is only true if the cross quadruples have certain properties. As explained above, a cross quadruple can be thought of as a quadruple unresolved by the input trees -- so in order for it to be resolved uniquely by all possible supertrees, there must be additional information on the quadruple somewhere in the input sets interplay. One way of providing such extra information will be given by the notion of so-called fixing taxa, which we now formally define. Note that we do not limit the definition to the important quadruple case, as a more general version will prove to be useful later on.

\begin{definition}[Fixing taxon] \label{def:fixtaxon} Let $X$ be a set of taxa and $\sS=\{Y_1,\ldots,Y_k\}$ be a set of subsets of $X$. Let $C=\{x_1,\ldots,x_c\}$ with $|C|=c$ be a  cross $c$-tuple of $\sS$.  Let $x \in X \setminus C$. Then, $x$ is called a \emph{fixing taxon of $C$ concerning $\sS$} if for all $i \in \{1,\ldots,c\}$ there is a $j \in \{1,\ldots,k\}$ such that  $\left(C \setminus\{x_i\} \right)\cup \{x\} \subseteq Y_j$.
\end{definition}

As cross quadruples and their fixing taxa are the main notion of the present manuscript, we now present an example to illustrate them.

\begin{example} \label{CQpossible2_part1} Let $X=\{1,2,3,4,5,6\}$ and $\sS:= \left\{ \{1,2,3,5\},\{1,2,4,5\}, \{1,2,4,6\},\{1,2,5,6\}\right.$,\\ $\left.\{1,2,3,6\},\{1,3,4,6\},\{1,3,5,6\}, \{1,4,5,6\},\{2,3,4,5\},\{2,3,5,6\},\{2,3,4,6\}\right\}$. In this case, $\sS$ has four CQs: $ \{1,2,3,4\},$ $\{1,3,4,5\},$ $\{2,4,5,6\}$ and $\{3,4,5,6\}.$ It can be easily verified that $ \{1,2,3,4\}$ has a fixing taxon, namely 6: This is true because all four quadruples $\{1,2,3,6\}$, $\{1,2,4,6\}$, $\{1,3,4,6\}$ and $\{2,3,4,6\}$ are contained in $\sS$. Similarly, $\{2,4,5,6\}$ has fixing taxon $1$. The other two CQs do not have any fixing taxa. 
\end{example}

We now show the role of fixing taxa in resolving cross quadruples.

\begin{proposition} \label{fixpointresolution} Let $X$ be a set of $n \geq 5$ taxa and $\sS=\{Y_1,\ldots,Y_k\}$ be a set of subsets of $X$. Let $\{a,b,c,d\}$ be a CQ of $\sS$ with fixing taxon $x \in X$. Then, any assignment of compatible trees $\T_1,\ldots,\T_k$ on the taxon sets $Y_1,\ldots,Y_k$ resolves $\{a,b,c,d\}$ in a unique way, i.e., for all pairs of supertrees $\T$, $\widetilde{\T}$of $\T_1,\ldots,\T_k$, we have: $\T\vert _{\{a,b,c,d\}}=\widetilde{\T}\vert _{\{a,b,c,d\}}$.
\end{proposition}

\begin{proof} Let $\{a,b,c,d\}$ be a CQ of $\sS$ with fixing taxon $x \in X$. Then, for any assignment of compatible trees $\T_1,\ldots,\T_k$ on taxon sets $Y_1,\ldots,Y_k$, each of the sets $\{a,b,c,x\}$, $\{a,b,d,x\}$, $\{a,c,d,x\}$ and $\{b,c,d,x\}$ is a subset of at least one of the sets $Y_i$ (for $i \in \{1,\ldots, k\}$) by Definition \ref{def:fixtaxon}. However, by Theorem \ref{thm_moan}, if we have at least $f_5=\binom{5}{4}-(5-4)=5-1=4$ quadruples on taxon set $\{a,b,c,d,x\}$, this set of quadruples is phylogenetically decisive for $\{a,b,c,d,x\}$. This implies that regardless of which compatible set of trees we choose for $Y_1,\ldots,Y_k$ and no matter if these trees have a unique supertree or not, at least their respective subtree on taxa $\{a,b,c,d,x\}$ is uniquely determined. In particular, this also shows that the subtree on $\{a,b,c,d\}$ is uniquely determined. More precisely, if $\T$ and $\widetilde{\T}$ are supertrees for $\T_1,\ldots,\T_k$, then we must have $\T\vert_{\{a,b,c,d\}}= \widetilde{\T}\vert_{\{a,b,c,d\}}$. This completes the proof. 
\end{proof}

\begin{example}[Example \ref{CQpossible} continued.] In Example \ref{CQpossible}, taxon $3$ is a fixing taxon of the only CQ $\{1,2,4,5\}$. Thus, by Proposition \ref{fixpointresolution}, this CQ is resolved in the same way in all possible supertrees of any particular choice of input trees corresponding to the taxon sets in $\sS$. 
\end{example}

 So by Proposition \ref{fixpointresolution}, the existence of a fixing taxon is sufficient  for a cross quadruple to be resolved. However, this condition is not necessary, which we illustrate with the following example, which shows that there are phylogenetically decisive collections of taxon sets that contain CQs which do not have any fixing taxa.

\begin{example}[Example \ref{CQpossible2_part1} continued] \label{CQpossible2} It can be checked that $\sS$ as given in Example \ref{CQpossible2_part1} is phylogenetically decisive, e.g., by verifying the four-way partition property (cf. Theorem \ref{sanderson_steel}) explicitly for all possible partitions of $X$, cf. Table \ref{tab_MainEx} in the Appendix.  However, as we have already seen in Example \ref{CQpossible2_part1}, $\sS$ has two CQs which do not have any fixing taxa.
\end{example}

We have seen in Proposition \ref{fixpointresolution} that fixing taxa are very helpful to resolve cross quadruples. In the following section, we will use this concept to finally define \emph{fixing taxon traceable} collections of taxon sets.

\subsection{Fixing taxon traceability and phylogenetic decisiveness } \label{sec:FTTintro}
We are now in the position to turn our attention to \emph{fixing taxon traceable}  collections of taxon sets. As we will show subsequently, these sets are guaranteed to be phylogenetically decisive and can be identified in polynomial time.  

The basis of the class of fixing taxon traceable collections of taxon sets is the idea that cross quadruples can be iteratively resolved: If a cross quadruple has no fixing taxon, but, say, taxon $x$ would be a fixing taxon if another quadruple was resolved which in turn does indeed have a fixing taxon, then the original cross quadruple can \enquote{inherit} its resolvability by resolving the second quadruple first. Thus, we need to distinguish directly and indirectly resolved cross quadruples, which leads to the following recursive definition, which is fundamental for this section. Note that we do not limit this definition to the quadruple case, as we will need this slightly more general version later on in this manuscript.

\begin{definition}[Fixing taxon traceability]\label{def:indirect} Let $X$ be a set of taxa and $\sS=\{Y_1,\ldots,Y_k\}$ be a set of subsets of $X$. Let $C=\{x_1,\ldots,x_c\}\subseteq X$ with $|C|=c\geq 1$ be a cross $c$-tuple of $\sS$. We call $C$ \emph{fixing taxon resolvable} if there is a taxon $x\in X \setminus C$ such that each of the $c$ sets $C_i:= \left(C\setminus \{x_i\}\right) \cup \{x\}$ (for $i=1,\ldots,c$) fulfills one of the following conditions: 
\begin{enumerate}
\item The set $C_i$ is not a cross $c$-tuple.
\item The set $C_i$ is a cross $c$-tuple but has a fixing taxon.
\item The set $C_i$ is a cross $c$-tuple but there exists  a sequence $C_1',\ldots,\C_l'$ of $c$-tuples of $X$ for some $l\in \mathbb{N}_{\geq 2}$ such that: \begin{itemize} \item $C_l'=C_i$, and \item $C_1'$ has a fixing taxon concerning $\sS$, and \item $C_j'$ has a fixing taxon concerning $\sS \cup \{C_1,',\ldots,C_{j-1}'\}$ for all $j=2,\ldots,l$. \end{itemize}

\end{enumerate} 
Moreover, $\sS$ is called \emph{fixing taxon $c$-traceable} if it either has no cross $c$-tuples or if all its cross $c$-tuples are fixing taxon resolvable. Finally, if $\sS$ is fixing taxon $4$-traceable, $\sS$  is called \emph{fixing taxon traceable} for simplicity. 
\end{definition} 

 The idea of Definition \ref{def:indirect} becomes apparent when looking at the problem the other way around: Given $\sS$, one can easily list all quadruples present in $\sS$, say in a list named \enquote{white}, and all CQs, say in a list named \enquote{gray}. Then, for each CQ it can be checked if it has a fixing taxon. Whenever a CQ is found which has a fixing taxon, it can be moved from the gray list to the white list. Now as the white list keeps growing, this may imply that a CQ which did not have a fixing taxon in the beginning as some quadruple was missing on the white list, at some stage gets a fixing taxon (as new quadruples are added to the white list). This also is the main idea behind Algorithm \ref{alg_generalFTT}, which we will state and analyze in Section \ref{sec_alg}. 
 
 We will now illustrate the notion of fixing taxon traceability with an example. 

\begin{example} [Examples \ref{CQpossible2_part1} and  \ref{CQpossible2} continued] \label{3overlapExampleFirstpart}  $X=\{1,2,3,4,5,6\}$ and $\sS:=$ \\$ \left\{ \{1,2,3,5\},\{1,2,4,5\}, \{1,2,4,6\},\{1,2,5,6\},\{1,2,3,6\},\{1,3,4,6\},\{1,3,5,6\},\{1,4,5,6\},\{2,3,4,5\}\right.$,\\ $\left.\{2,3,4,6\},\{2,3,5,6\}\right\}$. As stated before, $\sS$ has four CQs, two of which have a fixing taxon: 
\begin{itemize}
\item $ \{1,2,3,4\}$ with fixing taxon 6 and
\item $\{2,4,5,6\}$ with fixing taxon 1.
\end{itemize}
Moreover, for the other two CQs, we have: 

\begin{itemize}
\item For the CQ $\{1,3,4,5\}$, we consider taxon $c=2$: The sets $C_i$ induced by Definition \ref{def:indirect} then are: $C_1=\{2,3,4,5\}$, $C_2=\{1,2,4,5\}$, $C_3=\{1,2,3,5\}$ and $C_4=\{1,2,3,4\}$. While $C_1$, $C_2$ and $C_3$ are all no CQs as they are contained in $\sS$, we have already seen above that $C_4$ has a fixing taxon. So by Definition \ref{def:indirect}, the CQ $\{1,3,4,5\}$ is fixing taxon resolvable.

\item For the CQ $\{3,4,5,6\}$, we consider taxon $c=1$: The sets $C_i$ induced by Definition \ref{def:indirect} then are: $C_1=\{1,4,5,6\}$, $C_2=\{1,3,5,6\}$, $C_3=\{1,3,4,6\}$ and $C_4=\{1,3,4,5\}$. While $C_1$, $C_2$ and $C_3$ are all no CQs as they are contained in $\sS$, we have already seen above that $C_4$ is fixing taxon resolvable. So by Definition \ref{def:indirect}, the CQ $\{3,4,5,6\}$ is fixing taxon resolvable, too.
\end{itemize}

As all CQs of $\sS$ have a fixing taxon or are fixing taxon resolvable, $\sS$ is fixing taxon traceable by Definition \ref{def:indirect}.
\end{example}

Example \ref{3overlapExampleFirstpart} shows that the set $\sS:= \left\{ \{1,2,3,5\},\{1,2,4,5\}, \{1,2,4,6\},\{1,2,5,6\},\right.$\\ $\left.\{1,2,3,6\},\{1,3,4,6\},\{1,3,5,6\}, \{1,4,5,6\},\{2,3,4,5\},\{2,3,4,6\},\{2,3,5,6\}\right\}$ on taxon set \\$X=\{1,2,3,4,5,6\}$ is fixing taxon traceable, and we have already seen in Example \ref{CQpossible2} that $\sS$ is phylogenetically decisive. It is the main aim of this section to show that this is not a coincidence. 

In fact, we now state and prove the main result of the present section, which shows why fixing taxon traceable sets are so important: In the biologically relevant case of $c=4$, fixing taxon traceable collections of taxon sets can be guaranteed to be phylogenetically decisive.

\begin{theorem}\label{thm_FTTmain}
Let $X$ be a set of taxa and $\sS=\{Y_1,\ldots,Y_k\}$ be a set of subsets of $X$ which is fixing taxon traceable. Then, $\sS$ is also phylogenetically decisive.
\end{theorem}

\begin{proof} If $\sS$ has no CQs, there is nothing to show -- it is clear that in this case, $\sS$ is fixing taxon traceable by Definition \ref{def:indirect} and it is also phylogenetically decisive by Proposition \ref{phylodecisive2}. So we may assume from now on that $\sS$ does indeed have at least one  CQ. As $\sS$ is fixing taxon traceable, by definition we can order the CQs of $\sS$ such that the first one has a fixing taxon and can thus be considered \enquote{resolved}, and the subsequent ones eventually can be resolved iteratively, too, so they  either have a fixing taxon or are fixing taxon resolvable. At any point in the iteration, this means that if a given CQ is fixing taxon resolvable, we can apply Proposition \ref{fixpointresolution} and consider the CQ as \enquote{resolved}; or, in other words, add it to $\sS$. Note that we in this way extend $\sS$ without extending the set of quadruples it resolves -- in fact, we only add CQs which are at least fixing taxon resolvable.  Thus, $\sS$ gets larger and larger until it contains $\binom{|X|}{4}$ elements (as it is fixing taxon traceable, the process by definition does not stop before we have added \emph{all} possible quadruples on $X$ to $\sS$), but in each iteration, we only add a quadruple to $\sS$ that by Definition \ref{def:indirect} in that iteration has a fixing taxon and thus is by Proposition \ref{fixpointresolution} uniquely resolved by all supertrees of compatible input trees $\T_1,\ldots,\T_k$ on $Y_1,\ldots, Y_k$. This shows that $\sS$ is phylogenetically decisive and thus completes the proof.
\end{proof}
 
So by Theorem \ref{thm_FTTmain} we know that all fixing taxon traceable collections of taxon sets are also phylogenetically decisive. The opposite direction is, however, unfortunately not true. We now present an explicit example for a set $\sS$ of taxon sets which is phylogenetically decisive but \emph{not} fixing taxon traceable.

\begin{example}\label{badexample} Let $X=\{1,2,3,4,5,6\}$ and $\sS=\left\{\{1,2,3,5\},\{1,2,4,5\},\{1,2,4,6\},\{1,3,4,6\}\right.$,\\ $\left. \{1,3,5,6\}, \{1,4,5,6\}, \{2,3,4,5\}, \{2,3,4,6\}, \{2,3,5,6\}\right\}$. This set is indeed phylogenetically decisive, which can be verified with the help of Theorem \ref{sanderson_steel} (cf. Table \ref{tab_jannesEx} in the Appendix). However, this set has the following cross quadruples: $\{1, 2, 3, 4\}, \{1, 2, 3, 6\}, \{1, 2, 5, 6\}, \{1, 3, 4, 5\}, \{2, 4, 5, 6\}, \{3, 4, 5, 6\}$. None of these CQs has a fixing taxon, which is why $\sS$ is \emph{not} fixing taxon traceable. 
\end{example}

Theorem \ref{thm_FTTmain} together with Example \ref{badexample} highlights the fact that fixing taxon traceable collections of taxon sets are indeed a proper subset of all phylogenetically decisive collections of taxon sets. However, as we will show in the following section, they have a huge advantage: While deciding whether a given collection of taxon sets is phylogenetically decisive is a coNP-complete problem, we will now show that identifying fixing taxon traceable collections of taxon sets can be done in polynomial time.

 \subsection{\texorpdfstring{On the complexity of fixing taxon $c$-traceability}{On the complexity of fixing taxon c-traceability}}\label{sec_alg}

 In this section, we present a polynomial time algorithm which decides if a given set $\sS$ of taxon sets is fixing taxon traceable. In order to do so, we first need to introduce the concept of the \emph{3-overlap graph}, which we will use as a means to visualize cross quadruples and their iterative resolution in fixing taxon traceable collections of taxon sets. As before, we will introduce a more general version considering $c\geq 1$ rather than only the biologically relevant case of $c=4$, as we will need this later on.

\begin{definition}[(colored) $(c-1)$-overlap graph $K_{c+1}^n$ and the restriction $K_{c+1}^n \vert_{X'}$] Let $X$ be a set of $|X|=n$ taxa and let $\sS=\{Y_1,\ldots,Y_k\}$ be a set of subsets of $X$. Let $c \in \mathbb{N}_{\geq 1}$. Then, the \emph{$(c-1)$-overlap graph} $K_{c+1}^n=(V,E)$ of $n$ and $c$ is the graph which can be derived as follows: Its vertex set is $V=\binom{X}{c}$ (i.e., $V$ contains all possible $c$-tuples of $X$), and its edge set $E$ contains all pairs of vertices which overlap in $c-1$ taxa. We call $K_{c+1}^n$ \emph{colored} with respect to $\sS$, if all vertices corresponding to cross $c$-tuples of $\sS$ are colored gray and all other vertices are colored white. Moreover, in a slight abuse of notation, for $X' \subset X$ we define $K_{c+1}^n \vert_{X'}$ to be the restriction of $K_{c+1}^n$ to all vertices (and their induced edges) that correspond to $c$-tuples that are subsets of  $X'$. \end{definition}

\begin{remark} Note that the name of the $(c-1)$-overlap graph, namely $K^n_{c+1}$ with $c+1$ in the index rather than $c-1$, stems from the fact that its structure consists of many copies of the complete graph $K_{c+1}$ on $c+1$ vertices which are \enquote{glued together} in a certain way. We will elaborate on this structure later on.
\end{remark}

For our purposes, the case $c=4$, which leads to quadruples (4-tuples) and the 3-overlap graph $K_5^n$, is the most important case. 

\begin{example} [Examples \ref{CQpossible2_part1},   \ref{CQpossible2} and \ref{3overlapExampleFirstpart} continued] \label{3overlapExample}  $X=\{1,2,3,4,5,6\}$ and $\sS:=\left\{ \{1,2,3,5\}\right.$,\\ $\left.\{1,2,4,5\}, \{1,2,4,6\},\{1,2,5,6\},\{1,2,3,6\},\{1,3,4,6\},\{1,3,5,6\},\{1,4,5,6\},\{2,3,4,5\},\{2,3,4,6\}\right.$,\\ $\left.\{2,3,5,6\}\right\}$. As stated before, $\sS$ has four CQs: $ \{1,2,3,4\},$ $\{1,3,4,5\},$ $\{2,4,5,6\}$ and $\{3,4,5,6\}.$  We construct the colored 3-overlap graph $K_5^6$ as depicted by Figure \ref{beforeIT1}. Note that the CQs are all gray and all quadruples present in $\sS$ are white. 
\end{example}

We now state Algorithm \ref{alg_generalFTT}, which allows us to identify fixing $c$-traceable collections of taxon sets by explicitly checking for each white quadruple with a gray neighbor in the colored $(c-1)$-overlap graph $K_{c+1}^n$ if one of its taxa acts as a fixing taxon for said neighbor.\footnote{Note that Algorithm \ref{alg_generalFTT} uses the graph structure and graph coloring procedure we also use throughout this manuscript. However, this algorithm does not depend on the underlying graph structure; it only requires the underlying $c$-tuples. This is why we also state an alternative but equivalent version of Algorithm \ref{alg_generalFTT} in the appendix of this manuscript, cf. Algorithm \ref{alg_generalFTT_WOGRAPH} therein. The two algorithms only differ in lines 20--22.} 

\newpage 

{\small
\begin{algorithm}[H]
\caption{Fixing taxon $c$-traceability }\label{alg_generalFTT}
\LinesNumbered
 \SetKwInOut{Input}{Input}\SetKwInOut{Output}{Output}
 \vspace{0.15 cm}
 \Input{ $c \in \mathbb{N}_{\geq 1}$\\
 $n \in \mathbb{N}_{\geq c}$ \\
 set $\sS=\{Y_1,\ldots,Y_k\}: \ Y_i \subseteq X=\{1,\ldots,n\} \wedge \lvert Y_i\rvert \geq c 
 \ \forall i=1,\ldots,k$}
 \Output{\textsf{TRUE} if $\sS$ is fixing taxon $c$-traceable, \textsf{FALSE} else}
 \Init{}{
 $X \gets \{1,\ldots,n\}$\\
 \For{$i\gets 1$ \KwTo $\binom{n}{c}$}{
 $Q(i)\gets \mbox{$i^{th}$ $c$-tuple of $\binom{X}{c}$}$ \\
 $white(Q(i)) \gets 0$ }  
 $whiteCounter\gets 0$\\
 $newWhites \gets \emptyset$\\
}

\For{$i\gets 1$ \KwTo $ \binom{n}{c} $ }
{\For{$j \gets 1$ \KwTo $k$}{
\If{$Q(i) \subseteq Y_j$}{$white(Q(i))\gets 1$ \\ $newWhites \gets newWhites \cup \{Q(i)\}$\\
$whiteCounter \gets whiteCounter +1$\\
break\\ }
}}

\While{$whiteCounter < \binom{n}{c}$ \& $newWhites \neq \emptyset$}{

$tuple \gets newWhites(1)$\\
$X'\gets X \setminus tuple$ \\

\For{$i \gets 1$ \KwTo $|X'|=n-c$}{
$x \gets X'(i)$

\tcc*[f]{Next: Check if there are $c$ white vertices in the $K_{c+1}$ subgraph  employing numbers $tuple \cup \{x\}$ of graph $K_{c+1}^n$ }\\

$G \gets K_{c+1}^n\vert_{tuple \cup \{x\}}$\\
\If { $\sum\limits_{j: Q(j) \in G} 
 white(Q(j))==c$}{ \For{$j\gets 1$ \KwTo $|V(G)|$} {\If{ $white(Q(j))==0$}{$white(Q(j))=1$\\ $newWhites \gets newWhites \cup \{Q(j)\}$\\
 $whiteCounter \gets whiteCounter+1$\\ break}}
 }
}
$newWhites\gets newWhites\setminus\{newWhites(1)\}$\\
}
\If{$whiteCounter==\binom{n}{c}$}{\Return{$\mathsf{TRUE}$}}
\Else{\Return{$\mathsf{FALSE}$}}
\end{algorithm}
}

Before we continue with a run-time analysis of Algorithm \ref{alg_generalFTT}, we present a detailed example for the biologically relevant case where $c=4$, i.e., the case where we have quadruples and consider their 3-overlap graph.

\begin{example}\label{3overlapExampleNEW} We continue Examples \ref{CQpossible2_part1},  \ref{CQpossible2}, \ref{3overlapExample} and \ref{3overlapExampleFirstpart} with $X$ and $\sS$ as stated there. As we have already seen, $\sS$ has four CQs: $ \{1,2,3,4\},$ $\{1,3,4,5\},$ $\{2,4,5,6\}$ and $\{3,4,5,6\}$, and the colored 3-overlap graph is depicted by Figure \ref{beforeIT1}. Note that the CQs are all gray and all quadruples present in $\sS$ are white, so we have already completed the initialization step of the algorithm. 

\begin{figure}[ht]      \centering\vspace{0.5cm} 
      \includegraphics[width=10cm]{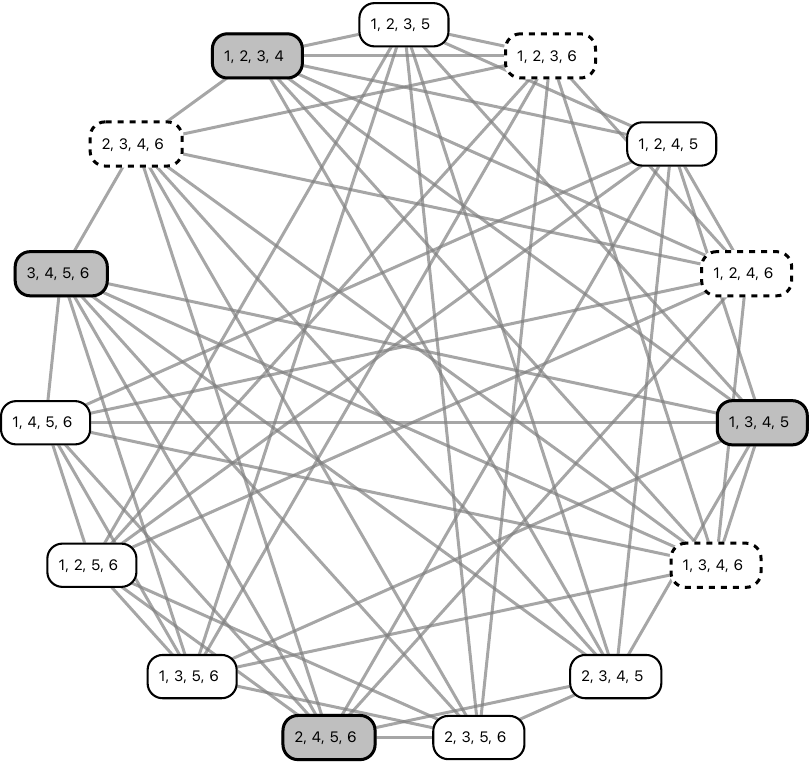}
    \caption{The 3-overlap graph for $n=6$ and $c=4$. The vertices have been colored to highlight the cross quadruples of the set $\sS:= \left\{ \{1,2,3,5\},\{1,2,4,5,\{1,2,4,6\},\{1,2,5,6\},\{1,2,3,6\}\}\right.$,\\ $\left. \{1,3,4,6\},\{1,3,5,6\}, \{1,4,5,6\},\{2,3,4,5\}\{2,3,5,6\},\{2,3,4,6\}\right\}$. The CQ $\{1,2,3,4\}$ is the one we first want to resolve using fixing taxon 6. The respective resolved neighbors employing taxon 6 are highlighted by dashed boxes. }\label{beforeIT1}
\end{figure}

 In lines 8--14 of Algorithm \ref{alg_generalFTT}, we make sure that all white nodes are regarded as newly colored white. The counter is set to the number of white nodes, which is 11. As this is less than the total number of quadruples of $X$, which is $\binom{6}{4}=15$, and as the set $newWhites$ of newly colored white nodes is not empty, we enter the \textsf{while}-loop. 

Next, we take the first newly colored white quadruple into account; without loss of generality we choose $\{1,2,3,5\}$. Now we have to consider the taxa in $X\setminus\{1,2,3,5\}=\{4,6\}$. Let us start with taxon $4$, i.e., we now want to find out if any of the gray neighbors of $\{1,2,3,5\}$ that contain taxon 4 can be colored white. Therefore, we construct graph $G$ as in line 20 of Algorithm \ref{alg_generalFTT}: We consider the restriction of $K_5^6$ to taxon set $\{1,2,3,4,5\}$ and check if this graph of five vertices (which is actually isomorphic to $K_5$, the complete graph on five vertices) contains four white vertices (in this case, the shared taxon of these four white vertices would act as a fixing taxon for the fifths quadruple in $G$). However, as both $\{1,2,3,4\}$ and  $\{1,3,4,5\}$ are gray, nothing gets re-colored.

Next, we consider taxon $6$, i.e., we now want to find out if any of the gray neighbors of  $\{1,2,3,5\}$ containing taxon 6 can be colored white. However, the restriction $K_5^6$ to taxon set $\{1,2,3,5,6\}$ already consists of five white vertices. So again, nothing gets re-colored. 

We delete $\{1,2,3,5\}$ from the list of newly white vertices and see that still not all nodes are white. Moreover, we still have newly white vertices which have not been analyzed yet, so we enter the \textsf{while}-loop once more. 

We now again take the first element of the set $newWhites$, without loss of generality $\{1,2,3,6\}$. Now we have to consider the taxa in $X\setminus\{1,2,3,6\}=\{4,5\}$. Let us start with taxon $4$, i.e., we now want to find out if any of the gray neighbors of  $\{1,2,3,6\}$ containing taxon 4 can be colored white. Again, we restrict $K_5^6$, this time to taxon set $\{1,2,3,4,6\}$, and this time see that indeed four vertices of this subgraph $G$ are white (namely $\{1,3,4,6\}$, $\{1,2,4,6\}$, $\{1,2,3,6\}$ and  $\{2,3,4,6\}$),  and one is gray (namely  $\{1,2,3,4\}$). So the latter quadruple now gets colored white and added to the newlyWhite list (as clearly, taxon 6 acts as a fixing taxon for this quadruple) in lines 24--25. The counter is also increased by 1 to 12 as we now in total have 11+1=12 white vertices.

Next, we consider taxon $5$, i.e., we now want to find out if any of the gray neighbors of $\{1,2,3,6\}$ containing taxon 5 can be colored white. However, as all vertices in the restriction of $K_5^6$ to $\{1,2,3,5,6\}$ are already white, nothing gets re-colored. This finishes the analysis of $\{1,2,3,6\}$. This quadruple thus gets removed from the list $newWhites$. However, the counter still has a value of $12<15$ and $newWhites\neq \emptyset$, so we start the \textsf{while}-loop yet again.

We repeat this procedure until finally all quadruples are colored white (which in this case is indeed possible as $\sS$ is fixing taxon traceable, cf. Example \ref{3overlapExampleFirstpart}). This implies that Algorithm \ref{alg_generalFTT} returns $TRUE$, so it does detect that $\sS$ is fixing taxon traceable.
\end{example}

Next, we want to prove that Algorithm \ref{alg_generalFTT} identifies all fixing taxon traceable collections of taxon sets in polynomial time. Again, we consider the general case of fixing taxon $c$-traceability.

\begin{proposition}\label{prop_runtime} 
 Given $c \in \mathbb{N}_{\geq 1}$ and a set $\sS=\{Y_1,\ldots,Y_k\}$ of subsets of $X$, where $\vert X\vert =n\geq c$, the question whether $\sS$ is fixing taxon $c$-traceable can be answered in at most  $\mathcal{O}\left(n^c\cdot \max\{k,n^2\} \right)$ steps using Algorithm \ref{alg_generalFTT}. In particular, as long as $k$ is polynomial in $n$, the question can be answered in polynomial time. 
\end{proposition}

\begin{proof} We first analyze the runtime of Algorithm \ref{alg_generalFTT}. 
The initialization step requires $\binom{n}{c}$ steps of constant time. Lines 8--14 simply list and count the potential quadruples on $X$ that are actually present in $\sS$ and should therefore be turned \enquote{white}. The two nested \textsf{for}-loops require $\binom{n}{c} \cdot k$ steps.

The most interesting part is the \textsf{while}-loop starting in line 15 of the algorithm. For each white $c$-tuple $C$ and each taxon $x$ of the remaining $\lvert X\rvert-c=n-c$ taxa, it considers $C$ together with all $c$-tuples that can be formed by joining a $(c-1)$-subset of $C$ with $\{x\}$ (these $c$-tuples form $G$, which is isomorphic to $K_{c+1}$) and sums up their \enquote{color status}, which is either 0 (gray) or 1 (white). This sum of $c+1$ summands (which are actually precisely the elements of the restriction of $K_{c+1}^n$ to the vertex set $C \cup \{x\}$) can be derived in constant time each from the list $white$, which means that calculating the sum takes $c+1$ constant time steps. 

Then, the \textsf{while}-loop calls the \textsf{for}-loop in line 18 of the algorithm, all of whose steps can be performed in constant time and which runs through all $c+1$ elements of $G$, so we have another $c+1$ constant time steps. Note that we repeat these in total $2(c+1)$ constant time steps at most $\binom{n}{c} \cdot (n-c)$ times, as the number of white $c$-tuples is bounded by $\binom{n}{c}$ and as we consider all $n-c$ remaining taxa. Therefore, in summary, the \textsf{while}-loop takes at most  $(2c+2)\cdot \binom{n}{c} \cdot (n-c) \leq (2n+2)\cdot \binom{n}{c} \cdot n$ steps (where the inequality uses $0 < c \leq n$).

In summary, the maximum number of steps required by Algorithm \ref{alg_generalFTT} is at most:

\begin{align*}\binom{n}{c} \cdot k+ (2n+2)\cdot \binom{n}{c} \cdot n &= \binom{n}{c} \cdot \left( k+2n^2+2n\right).\end{align*}

This gives a total complexity of $\mathcal{O}\left(n^c\cdot \max\{k,n^2\} \right)$, which completes the first part of the proof.

It remains to show that Algorithm \ref{alg_generalFTT} returns \textsf{TRUE} if and only if the input set $\sS$ is fixing taxon $c$-traceable, but this is simple as the algorithm basically only iteratively goes through Definition \ref{def:indirect}. In particular, as Algorithm \ref{alg_generalFTT} exhaustively checks for each of the quadruples in $\sS$ if this quadruple helps together with \emph{any} taxon to fix \emph{any} cross quadruple, it is clear that no quadruple that has a fixing taxon can be missed. For every such quadruple then, in turn, it is again checked if it helps together with \emph{any} taxon to fix \emph{any} cross quadruple. Thus, it is clear that all fixing taxon resolvable quadruples are found, too. 
On the other hand, however, Algorithm \ref{alg_generalFTT} cannot produce any \enquote{false positives}, either, as the crucial checks in lines 21--26, which decide if a CQ gets colored white, are merely checks if taxon $x$ is a fixing taxon for the current tuple under investigation; i.e., this test is simply a realization of Definition \ref{def:fixtaxon}. This observation completes the proof.
\end{proof}

\begin{remark} It can easily be seen that other approaches than Algorithm \ref{alg_generalFTT} would also lead to a solution and -- depending on the input set $\sS$ -- possibly to faster run times. For instance, instead of checking for all white tuples if they help to re-color gray tuples, one could consider all gray tuples and check if they can be re-colored using the existing white tuples. Similarly, one could check all taxa and see if they act as fixing taxa for any gray tuple. However, we refrain here from investigating these modified algorithms further, because the aim of this section is only to show that there is a way of checking fixing taxon $c$-traceability in polynomial time, not to investigate several possible approaches.
\end{remark}

Now that we know that fixing taxon traceable collections of taxon sets can be identified in polynomial time, we can turn our attention to another important question: How many quadruples are needed in a set $\sS$ of taxon sets to have a chance to be a) fixing taxon traceable or b) phylogenetically decisive?

\subsection{Bounds on the number of \texorpdfstring{$c$}{c}-tuples needed for fixing taxon \texorpdfstring{$c$}{c}-traceability}\label{sec_FTTbounds}

In the previous sections we have seen that fixing taxon traceable collections of taxon sets form a subset of all phylogenetically decisive collections of taxon sets, and that -- while the latter are hard to identify -- the former can be identified in polynomial time. However, as the two sets are different (which we will elaborate on more in-depth in Section \ref{sec_diff}), the questions of how many quadruples are actually needed in a set $\sS$ of taxon sets at least to have a chance that it is fixing taxon traceable or phylogenetically decisive or even to be sure that it fulfills either of these properties have to be answered separately.

We start with the following bound for guaranteed fixing taxon $c$-traceability.

\begin{theorem}\label{thm_upperbound}
Let $c\in \mathbb{N}_{\geq 1}$. Let $\sS$ be a set of subsets of taxon set $X$ with $|X|=n\geq c$. Let $k$ be the number of $c$-tuples that are subsets of the elements of $\sS$, i.e., there are $\binom{n}{c}-k$ cross $c$-tuples  induced by $\sS$. Then, if $k\geq \binom{n}{c}-n+c$, $\sS$ is fixing taxon $c$-traceable. Additionally, if $c=4$, $\sS$ is phylogenetically decisive.
\end{theorem}

Before we prove the theorem, we note that Theorem \ref{thm_upperbound} together with Theorem \ref{thm_FTTmain} provides an alternative proof of Theorem \ref{thm_moan}, because (by Theorem \ref{thm_upperbound}) $k\geq \binom{n}{4}-n+4$ implies fixing taxon traceability and (by Theorem \ref{thm_FTTmain}) all fixing taxon traceable sets are phylogenetically decisive.

\begin{proof}[Proof of Theorem \ref{thm_upperbound}] 
Concerning the first assertion, we show that if $k\geq \binom{n}{c}-n+c$, $\sS$ is fixing taxon $c$-traceable. If $c=4$, phylogenetical decisiveness then follows by Theorem \ref{thm_FTTmain}.

 \begin{figure}[ht]      \centering\vspace{0.5cm} 
      \includegraphics[width=5.5cm]{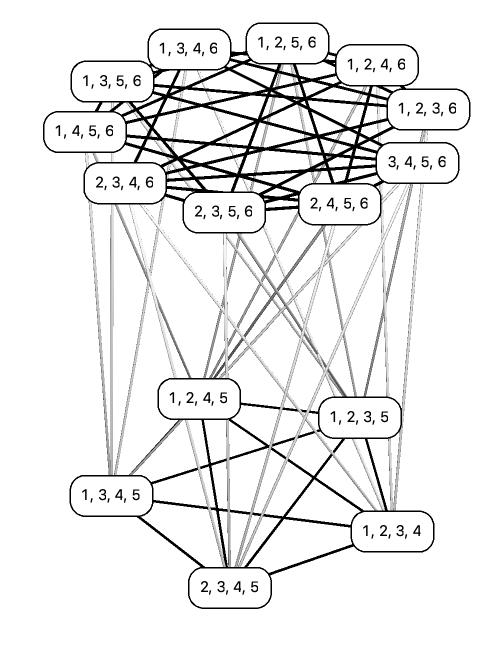} \hspace{0.5cm}
        \includegraphics[width=5.5cm]{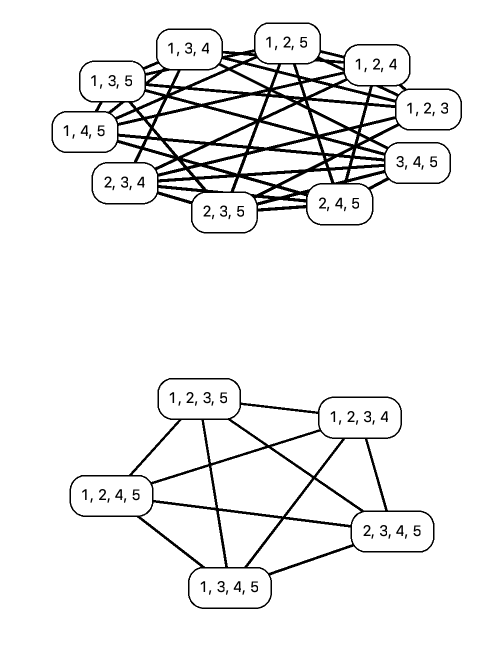} 
    \caption{A depiction of the graph $K_5^6$ (left) as well as $K_4^5$ (top right) and $K_5^5$ (bottom right). Note that $K_5^5$ is isomorphic to $K_5$, and note that $K_5^6$ contains two subgraphs which are isomorphic to $K_4^5$ and $K_5^5$, respectively. In the proof of Theorem \ref{thm_generaloverlapcoloring}, $K_5^5$ corresponds to $G_1$ and $K_4^5$ corresponds to $G_2$.}\label{structureKnc}
  \end{figure}

Let $\widetilde{k}:=\binom{n}{c}-k$ denote the number of cross $c$-tuples induced by $\sS$. Then, we note that every cross $c$-tuples has at most $\widetilde{k}-1$ neighboring cross $c$-tuples in the colored $(c-1)$-overlap graph $K_{c+1}^n$. 

A second fact that we note is that if a cross $c$-tuple has fewer than $n-c$ gray neighbors in $K_{c+1}^n$, i.e., if fewer than $n-c$ cross $c$-tuples are neighbors of the original cross $c$-tuple, then this cross $c$-tuple must have a fixing taxon. To see this, first recall that each $c$-tuple $C$ for each taxon $x \in X\setminus C$ has $c$ neighbors in $K_{c+1}^n$  which contain $c-1$ taxa of $C$ together with $x$. As $\lvert X\setminus C\rvert=n-c$, we know that each $c$-tuple has $c(n-c)$ neighbors in $K_{c+1}^n$. For each $x$, the subgraph consisting of the $c$-tuples employing taxa $C\cup \{x\}$ is a $K_{c+1}$, i.e., a complete graph with $c+1$ vertices. Now, if there is \emph{no} fixing taxon for a cross $c$-tuple $C$, this implies that in all these $K_{c+1}$-subgraphs to which $C$ belongs, one more $c$-tuple (other than $C$) is a cross $c$-tuple. 
So if a cross $c$-tuple has fewer than $n-c$ gray neighbors in $K_{c+1}^n$, at least one of these $(n-c)$ many $K_{c+1}$-subgraphs must be all white except for $C$. The unique taxon used in the vertices of this $K_{c+1}$ which is not contained in $C$ is thus a fixing taxon for $C$.

This shows that if we can guarantee that a cross $c$-tuple has fewer than $n-c$ gray neighbors in the colored $(c-1)$-overlap graph $K_{c+1}^n$, we know that it can be resolved. So if we have that the maximum possible number $\widetilde{k}-1$ of gray neighbors of any cross $c$-tuple, or -- in other words -- one less than the total number of cross $c$-tuples of $\sS$, is less than $n-c$, we can thus even guarantee that \emph{every} cross $c$-tuple has a fixing taxon. However, by assumption we have $k\geq \binom{n}{c}-n+c$. This implies $\binom{n}{c}-k \leq n-c$, which (using $\widetilde{k}=\binom{n}{c}-k$) shows that $\widetilde{k}\leq n-c$ and thus also $\widetilde{k}-1< n-c$. So indeed there cannot be any cross $c$-tuple that does not have a fixing taxon. This completes the proof.

\end{proof}

\begin{remark} Before we continue, we point out that it was shown by \citet[Theorem 3.3]{moan} that for all  $k\leq \binom{n}{4}-n+3$, there are sets $\sS$ with $k$ quadruples which are \emph{not} phylogenetically decisive. Thus, these sets (by Theorem \ref{thm_FTTmain}) are also not fixing taxon traceable. This implies that the bound provided by Theorem \ref{thm_upperbound} is sharp both for phylogenetic decisiveness and fixing taxon traceability (for $c=4$).
\end{remark}

While for the biologically relevant case of $c=4$, Theorem \ref{thm_upperbound} provides a minimal number $k$ of quadruples for which both fixing taxon traceability and phylogenetic decisiveness can be guaranteed, the question of how many quadruples are needed so that a set $\sS$ consisting of these quadruples \emph{can} be phylogenetically decisive or fixing taxon traceable has a different flavor. Concerning phylogenetic decisiveness, this question has been tackled in the literature, but as we will see later on, the lower bound of $\binom{n-1}{3}$ on the number of quadruples needed for a $\sS$ on taxon set $X=\{1,\ldots,n\}$ to be phylogenetically decisive suggested by \citet[Theorem 4(i)]{parvini2022} is unfortunately incorrect. In fact, we will show how to construct phylogenetically decisive collections of taxon sets employing precisely $\binom{n-1}{3}-1$ quadruples. Remarkably, the erroneous lower bound for decisiveness turns out to be the correct and sharp lower bound for fixing taxon traceability. This is a direct consequence of the following theorem, which establishes a lower bound for the number of quadruples required in a set $\sS$ of taxon sets in order for $\sS$ to be fixing taxon traceable.

\begin{theorem}
\label{thm_generaloverlapcoloring} Let $n$, $c \in \mathbb{N}_{\geq 1}$: $n\geq c \geq 1$. Then, the following two statements hold: 
\begin{enumerate}
    \item There is a set $\sS$ of $c$-tuples of the set $\{1,\ldots,n\}$ with $|\sS|= \binom{n-1}{c-1}$ such that $\sS$ is fixing taxon $c$-traceable.    \item Let $\sS$ be a set of $c$-tuples of the set $\{1,\ldots,n\}$ which is fixing taxon $c$-traceable. Then, $|\sS|\geq \binom{n-1}{c-1}$.
\end{enumerate}
\end{theorem}

Before we can prove Theorem \ref{thm_generaloverlapcoloring}, we need to prove the following lemma.

\begin{lemma}\label{lem_Knc_connected} Let $n$, $c \in \mathbb{N}_{\geq 1}: n \geq c$. Then, the $(c-1)$-overlap graph $K_{c+1}^n$ of $n$ is connected.
\end{lemma}

\begin{proof} Note that by definition, $K_{c+1}^n$ has $\binom{n}{c}$ vertices, all of which are $c$-tuples. If $n=c$, $K_{c+1}^n$ contains only one vertex and there is nothing to show. Similarly, if $c=1$, there are $n$ vertices, which correspond to the singleton sets $\{1\},\ldots, \{n\}$, and which are all connected with one another in the 0-overlap graph (i.e., the form a complete graph $K_n$) as they have a 0 overlap, so there, too, is nothing left to show. Therefore, we now consider the case $n>c>1$, in which case we have $\binom{n}{c}>1$ vertices. So let $u=(u_1,\ldots,u_{c})$ and $v=(v_1,\ldots,v_{c})$ be vertices in $K_{c+1}^n$. We now argue why there exists a path from $u$ to $v$ in $K_{c+1}^n$.

If $|u \cap v|\geq c-1 $, either $u=v$ (if $|u \cap v|=c$) or there is an edge $e=\{u,v\}$ contained in $K_{c+1}^n$ (if $|u \cap v|=c-1$), so there is nothing to show. So assume  $|u \cap v|< c-1$. In this case, we now show that there is a neighbor $w$ of $u$ in $K_{c+1}^n$ with $|w \cap v|=|u \cap v|+1 $, so $w$ is \enquote{one step closer} to $v$ than $u$. In order to find $w$, let $a \in u \setminus v$ and $b \in v \setminus u$, and set $w:=\left(u \setminus \{a\}\right) \cup \{b\}$. Obviously, this gives $|w \cap v|=|u \cap v|+1 $, as $w$ contains one more element of $v$ than $u$, namely $b$. We can proceed in the same manner, i.e., we can always find a neighbor of the current node that is one step closer to $v$ as long as its intersection with $v$ contains less than $c-1$ vertices. This shows that we can iteratively find a  path from $u$ to $v$ in $K_{c+1}^n$ and thus completes the proof. 
\end{proof}

Before we now finally proceed to prove Theorem \ref{thm_generaloverlapcoloring}, we note that the recursive idea of the second part of the proof is somewhat related to the underlying recursive approach of the proof of \cite[Theorem 4(ii)]{parvini2022}, combined with the added idea of colored overlap graphs. This shows that the basic idea of  \cite[Theorem 4(ii)]{parvini2022} is still relevant, even though the result is unfortunately erroneous (as its bound is stated for all phylogenetically decisive collections of taxon sets rather than only fixing taxon traceable ones).

\begin{proof}[Proof of Theorem \ref{thm_generaloverlapcoloring}] $\mbox{}$
\begin{enumerate}
    \item We construct a collection of taxon sets $\sS$ consisting of the following $c$-tuples: We select all $v \in V\left( K_{c+1}^n\right)$ such that $v$ contains taxon $n$ as well as one of the $(c-1)$-tuples of $\{1,\ldots,n-1\}$. This set $\sS$ of vertices contains $\binom{n-1}{c-1}$ many different $c$-tuples, all of which contain taxon $n$. We now show that $\sS$ is fixing taxon $c$-traceable. 
    
    The main idea of this proof is that taxon $n$ will act as a fixing taxon for all $c$-tuples not containing $n$, i.e., for all cross $c$-tuples.
    
    In order to see this, let $u=(u_1,\ldots,u_{c}) \in V\left(K_{c+1}^n\right)$. Then, if $u\in \sS$, there is nothing to show. So let us consider the case $u \not\in \sS$. By definition of $\sS$, we have $n \not\in u$. We restrict $K_{c+1}^n$ to taxon set $\{u_1,\ldots,u_c,n\}$, i.e., we consider $G:={K_{c+1}^n}\vert_{\{u_1,\ldots,u_c,n\}}$. Clearly, $u \in G$. Moreover, by definition $u$ is the only vertex in $G$ which does not contain $n$. Thus, as all vertices of $G$ except for $u$ are contained in $\sS$ (as they all contain $n$). Thus, we have found $c$ neighbors of $u$, each containing $c-1$ elements of $u$ and $n$. This shows that $n$ is a fixing taxon for $u$, so $u$ gets resolved by Proposition \ref{fixpointresolution}.
    As $u$ was arbitrary, this shows that \emph{all} cross $c$-tuples of $\sS$ can get resolved, which completes the proof of the first statement of Theorem \ref{thm_generaloverlapcoloring}.
\item
We prove the second statement of Theorem \ref{thm_generaloverlapcoloring} for all pairs $(c,n) \in \mathbb{N}^2$ with $n\geq c \geq 1$ by double induction. Our proof strategy is as follows: First, we prove the statement for the case $(1,n)$, i.e., we prove that for $c=1$ and all $n\geq 1$, the statement holds. Then, in the inductive step $c\rightarrow c+1$, we perform an induction on $n$, starting with the base case $n=c+1$ and then going, in the inductive step, from $n$ to $n+1$. This final step is the most complex, but also the most beautiful one, as we exploit the structure of $K_{c+1}^{n+1}$ by showing that it can be subdivided into two induced subgraphs, one of which is isomorphic to $K_{c}^{n}$ and one of which is isomorphic to $K_{c+1}^{n}$. For these two subgraphs, we can use the respective inductive assumptions and simply add their values to get the desired lower bound for $K_{c+1}^{n+1}$. 

We now elaborate on the details of the proof. 
\begin{itemize}
\item $(1,n)$: For $c=1$, we have $\binom{n-1}{c-1}=\binom{n-1}{0}=1$, so we need to show that one single $c$-tuple (in this case a set containing only one taxon, as $c=1$) in $\sS$ suffices to color $K_2^n$ white, i.e., to make $\sS$ fixing taxon 1-traceable. However, note that we consider the 0-overlap graph $K_2^n$, which by Lemma \ref{lem_Knc_connected} is connected, and that fixing taxon 1-traceability is concerned with $K_2^n$'s  decomposition into $K_2$-graphs, which are simply single edges. As the graph is connected, this implies that if we color one vertex white, this will color all its neighboring vertices white, which in turn will color all their neighbors white and so forth. Ultimately, this allows us (by Algorithm \ref{alg_generalFTT}) to color the last vertex of $K_2$ white, too. This proves that indeed, if $c=1$ and if $n\geq c$, $\binom{n-1}{c-1}=\binom{n-1}{0}=1$ white vertex suffices to color $K_2^n$ white, i.e., one 1-tuple suffices to make a set $\sS$ fixing taxon 1-traceable. 
     
    \item $(c,n)\rightarrow (c+1,n)$: We now want to show that if the second statement of the theorem holds for the pair $(c,n)$ and if $n\geq c+1$, then it also holds for the pair $(c+1,n)$. We prove this by an induction on $n$, starting with the base case $n=c+1$. So in the following, we assume the statement is already shown to be true for $c$ and all values of $n\geq c$ (we call this \textbf{assumption $\bm{\mathcal{A}}$}), and we now want to show it for $c+1$ and all values of $n\geq c+1$. 
    \begin{itemize}
        \item $(c+1,n=c+1)$: We consider the case of $(c+1)$-tuples and $n=c+1$. Note that when $n=c+1$, i.e., the number of taxa equals the tuple size, the $c$-overlap graph $K_{c+2}^n=K_{c+2}^{c+1}$ contains only one vertex, which corresponds to the $(c+1)$-tuple $\{1,\ldots,c+1\}$. If this is a cross $(c+1)$-tuple, i.e., if this only available tuple is not contained in $\sS$, $\sS$ is obviously not fixing taxon $(c+1)$-traceable, because there are no fixing taxa for this cross tuple. Thus, this tuple must be contained in $\sS$, leading to:

        $$ \lvert \sS \rvert \geq 1 = \binom{c}{c}= \binom{n-1}{(c+1)-1}, $$
        which completes the proof for the case $(c+1,n=c+1)$.
       
            \item $(c+1,n)\rightarrow (c+1,n+1)$: Last, we assume that the statement is already proven to be true for the pair $(c+1,n)$ (and we call this \textbf{assumption} $\bm{\mathcal{B}}$) and want to prove that it then also holds for the pair $(c+1,n+1)$. The highlevel idea of this part of the proof is to exploit the structure of $K_{c+2}^{n+1}$ and decompose it into two subgraphs, one of which is isomorphic to $K_{c+2}^{n}$ (which will allow us to use assumption $\mathcal{A}$) and one of which is isomorphic to $K_{c+1}^{n}$ (which will allow us to use assumption $\mathcal{B}$). Basically, this will give us two binomial coefficients whose sum equals the desired term for $K_{c+2}^{n+1}$.
        
        To see this, note that if we consider $K_{c+2}^{n+1}$, it contains $\binom{n+1}{c+1}$ vertices corresponding to all $(c+1)$-tuples that can be formed from taxa in $\{1,\ldots,n+1\}$, and $\binom{n}{c+1}$ of them do \emph{not} contain taxon $n+1$, whereas $\binom{n}{c}$ of them do contain taxon $n+1$. We now analyze these two sets of vertices separately.
        \begin{itemize}
            \item Let us denote the subgraph of $K_{c+2}^{n+1}$ induced by the $\binom{n}{c+1}$ vertices which do \emph{not} contain taxon $n+1$ by $G_1$. Note that $G_1\cong K_{c+2}^{n}$ (cf. Figure \ref{structureKnc}). So in fact, by assumption $\mathcal{B}$, we need at least $\binom{n-1}{c}$ white vertices  in $G_1$ to color all of them white. Clearly, in order to color all vertices of $K_{c+2}^{n+1}$ white, its subgraph $G_1$ needs to be colored white, so this gives at least $\binom{n-1}{c}$ vertices, all of which consist of $(c+1)$-tuples \emph{not} containing taxon $n+1$.

            \item Now, if all vertices of $K_{c+2}^{n+1}$ which do \emph{not} contain taxon $n+1$ are colored white, those vertices which \emph{do} contain taxon $n+1$, of which there are $\binom{n}{c}$ many, still need to be colored white.

            We denote by $G_2$ the subgraph of $K_{c+2}^{n+1}$ induced by these vertices (cf. Figure \ref{structureKnc}). Let $v$ be a vertex in $G_2$. We now define a map $f:G_2 \rightarrow K_{c+1}^n$ as follows: $f(v):=v \setminus \{n+1\}$. Clearly, $f$ is a bijection (with $f^{-1}=v \cup \{n+1\}$). We now argue that $f$ is even a graph isomorphism. Let $e=\{u,v\}$ be an edge in $G_2$ and thus also in $K_{c+2}^{n+1}$ (as $G_2$ is an induced subgraph). Note that by definition of $K_{c+2}^{n+1}$, $u$ and $v$ correspond to $(c+1)$-tuples, and as they are connected by edge $e=\{u,v\}$, these tuples must overlap in precisely $c$ taxa, one of which is taxon $n+1$ (as $u$ and $v$ are vertices in $G_2$). So $f(u)$ and $f(v)$ are $c$-tuples that overlap in precisely $c-1$ taxa. This implies that $f(u)$ and $f(v)$ must be connected by an edge in $K_{c+1}^{n}$, so the edge $\left\{f(u),f(v)\right\}$ is contained in said graph. 
            
            On the other hand, if two vertices $a$ and $b$ are connected by an edge $e'$ in $K_{c+1}^{n}$, this means the corresponding $c$-tuples overlap in $c-1$ taxa, which means that $f^{-1}(a)=a \cup \{n+1\}$ and $f^{-1}(b)=b \cup \{n+1\}$ are $(c+1)$-tuples which overlap in $c$ positions (because they also overlap in taxon $n+1$). Thus, if there is an edge $e'=\{a,b\}$ in $K_{c+1}^{n}$, there also is an edge $\left\{f^{-1}(a),f^{-1}(b)\right\}$ in $K_{c+2}^{n+1}$ and thus (as $f^{-1}(a)$ and $f^{-1}(b)$ are both vertices in $G_2$ as they both contain taxon $n+1$) also in $G_2$. Therefore, $f$ is indeed a graph isomorphism.
            
            This in turn shows that $G_2$ is isomorphic to $K_{c+1}^n$, which implies that we already know by assumption $\mathcal{A}$ that we need at least $\binom{n-1}{c-1}$ white vertices in $G_2$ so that Algorithm \ref{alg_generalFTT} will turn all vertices of $G_2$ white. 
        \end{itemize}
        As $G_1$ and $G_2$ have no vertices in common, the number of vertices needed to be colored white in order to turn all vertices of $K_{c+2}^{n+1}$ white is simply the sum of the white vertices needed for $G_1$ and $G_2$, respectively. This immediately implies for the number $\widetilde{k}$ of $(c+1)$-tuples required for a set $\sS$ of such tuples to be fixing taxon $(c+1)$-traceable that
        
        $$\widetilde{k} \geq \underbrace{\binom{n-1}{c}}_{\mbox{due to $G_1$}} + \underbrace{\binom{n-1}{c-1}}_{\mbox{due to $G_2$}}=\binom{n}{c}=\binom{(n+1)-1}{(c+1)-1}.$$ Here, the first equality is a well-known identity from combinatorics, which is due to the tetrahedral numbers, and the second equality is the desired statement in the inductive step. This completes the proof.
    \end{itemize}
\end{itemize}

\end{enumerate}

\end{proof}

Theorem \ref{thm_generaloverlapcoloring} has important implications on fixing taxon traceability in the biologically important quadruple case, as we will state in the following corollary. The first part of this corollary gives us the desired lower bound for fixing taxon traceability in the quadruple setting.

\begin{corollary} \label{cor_lowerbound} Let $n \in \mathbb{N}_{\geq 4}$ and $X=\{1,\ldots,n\}$. Then, we have:
\begin{enumerate}
\item Let $\sS$ be a collection of taxon sets of $X$ containing $k$ quadruples. Then, if $\sS$ is fixing taxon traceable, $\sS$ induces at least $ \binom{n-1}{3}$ quadruples.
\item Let $k \geq \binom{n-1}{3}$. Then, there exists a fixing taxon traceable set $\sS$ of taxon sets inducing precisely $k$ quadruples. 
\end{enumerate}
\end{corollary}

\begin{proof} \mbox{}
\begin{enumerate} \item This is a direct consequence of the second part of Theorem \ref{thm_generaloverlapcoloring}, which can be derived by setting $c=4$. 
\item Using the first part of Theorem \ref{thm_generaloverlapcoloring}, and the construction given in the proof thereof, we easily see that the set $\sS$ consisting of all $\binom{n-1}{3}$ quadruples that contain taxon $n$ is fixing taxon traceable (with fixing taxon $n$ for all cross quadruples). Clearly, each set $\sS'$ with $\sS \subseteq \sS'$ is then also fixing taxon traceable, which proves the assertion.
\end{enumerate}
\end{proof}

Before we conclude this section, we digress from the biologically relevant case of $c=4$ to have a glimpse at the rooted tree setting, or, in other words, the case $c=3$.

\begin{remark} Consider $n\geq c=3$ and all $\binom{n}{c}=\binom{n}{3}$ possible $3$-tuples (triples) on $n$ taxa. Then, Theorem \ref{thm_generaloverlapcoloring} shows that in order for a collection of taxon sets $\sS$ to be fixing taxon 3-traceable, we require $\sS$ to contain at least $\binom{n-1}{c-1}=\binom{n-1}{2}<\binom{n}{3}$ many triples, and that there are indeed such sets $\sS$ with fewer than $\binom{n}{3}$ triples that are fixing taxon 3-traceable. As an example, consider $\sS$ as given in Example \ref{ex_triple}. 

However, this shows that interestingly, Theorem \ref{thm_FTTmain} does \emph{not} hold when rooted trees are considered instead of unrooted trees: In the rooted setting, we can obviously have fixing taxon 3-traceability without phylogenetic decisiveness (as it was shown by \cite[Theorem 2]{parvini2022} that in the rooted case, phylogenetic decisiveness requires the presence of \emph{all} $\binom{n}{3}$ possible triples). In particular, this implies that the fundamental Theorem \ref{thm_FTTmain}, which deals with the biologically relevant case of $c=4$ and unrooted trees, does not generelly hold for all values of $c$, as it already fails for $c=3$. \end{remark}

\begin{example}\label{ex_triple} Consider $X=\{1,2,3,4\}$ and the set $\sS=\{\{1,2,4\}, \{1,3,4\}$, $\{2,3,4\}\}$. So of the four triples possible for $X$, one is not present in $\sS$, namely $\{1,2,3\}$. So $\sS$ is \emph{not} phylogenetically decisive in the rooted sense by \cite[Theorem 2]{parvini2022}. However, $\sS$ is fixing taxon 3-traceable, because the four triples $\{1,2,4\}$, $\{1,3,4\}$, $\{2,3,4\}$ and $\{1,2,3\}$ form $K_4^4$, which is isomorphic to $K_4$. So only one vertex in this $K_4$ is not in $\sS$ and thus \emph{not} white in the beginning of Algorithm \ref{alg_generalFTT}. Thus, the missing triple $\{1,2,3\}$ will also get colored white in line 24 of the algorithm. So $\sS$ is fixing taxon 3-traceable, but it is not phylogenetically decisive in the rooted sense. 
\end{example}

So as the above example shows, in the rooted setting, fixing taxon 3-traceability does not help concerning the decision if a collection of taxon sets $\sS$ is phylogenetically decisive, but fortunately this decision consists only of checking whether all possible triples are present in $\sS$ \cite[Theorem 2]{parvini2022}, which can be done in polynomial time, anyway.

Next, we turn our attention to a new bound on the number of quadruples needed in order for a set $\sS$ of taxon sets to be phylogenetically decisive.

\subsection{On the number of quadruples needed for phylogenetic decisiveness}\label{sec_PDbounds}

We have seen in the previous section (Corollary \ref{cor_lowerbound}) that if a set $\sS$ of taxon sets contains fewer than $\binom{n-1}{3}$ quadruples, $\sS$ cannot be fixing taxon traceable. However, such a set $\sS$ might still be phylogenetically decisive as we have seen in Example \ref{badexample}: There, we have $n=6$, and $\sS$ contains $9<\binom{5}{3}=10$ quadruples, and $\sS$ fulfills the four-way partition property. This is a contradiction to  \cite[Theorem 4(i)]{parvini2022}, in which the authors state that all phylogenetically decisive collections of taxon sets contain at least $\binom{n-1}{3}$ quadruples.

It is therefore the main aim of the present section to present a corrected lower bound for the number of quadruples needed to make $\sS$ phylogenetically decisive. We do this with the following theorem.

\begin{theorem}\label{thm_lowerboundphylodec}
Let $n \in \mathbb{N}_{\geq 6}$. Let $\widehat{k}(n)$ denote the minimal number of quadruples present in any phylogenetically decisive collection of taxon sets on $X=\{1,\ldots,n\}$, i.e., there is no set $\mathcal{S}$ of subsets of $X$ that contains fewer than $\widehat{k}(n)$ quadruples and is phylogenetically decisive. Then, we have: 

 $$ \left\lceil\frac{1}{4}\ \binom{n}{3} \right\rceil +2\leq \widehat{k}(n) \leq \binom{n-1}{3}-1.$$
\end{theorem}

\begin{proof} 
Concerning the first inequality, note that in order for $\sS$ to be phylogenetically decisive, all partitions of the kind $a|b|c|X\setminus\{a,b,c\}$ need to be covered. As already observed in \cite{sanderson_steel_2010}, this implies that all triples $\{a,b,c\}$ must be a subset of at least one quadruple of $\sS$. Let $\sS$ be a phylogenetically decisive collection of taxon sets on $X=\{1,\ldots,n\}$ containing $k$ quadruples, and note that every quadruple contains four triples. This implies that $\sS$ contains $4k$ triples (possible duplicates included). However, there are $\binom{n}{3}$ possible triples, which all need to occur at least once. This immediately leads to: $4k \geq \binom{n}{3}$, and thus $k \geq \frac{1}{4}\ \binom{n}{3}. $ 

However, if $k=\frac{1}{4}\ \binom{n}{3}$, i.e., if every triple is covered precisely once\footnote{Note that in \cite{hanani1960} a necessary and sufficient condition on $n$ can be found for the existence of quadruple sets that cover each possible triple precisely once, namely if $n \equiv 2 \mbox{ or } 4 \ (\hspace{-0.2cm}\mod 6).$}, as $n\geq 6$, $\sS$ cannot be decisive. To see this, assume the quadruple $\{a,b,c,d\}$ is the only quadruple to cover both triples $\{a,b,c\}$ and $\{a,b,d\}$. Then, the four-way partition property is violated as the partition $X_1=\{a\}$, $X_2=\{b\}$, $X_3=\{c,d\}$ and $X_4=X \setminus \{a,b,c,d\}$ is not covered (cf. Theorem \ref{sanderson_steel}). Thus, we require $k\geq \frac{1}{4}\ \binom{n}{3}+1$. But this is still not sufficient. This is due to the fact that if only one of the mentioned triples, say $\{a,b,c\}$ appears in a second quadruple, say $\{a,b,c,e\}$, then the four-way partition $X_1=\{a\}$, $X_2=\{b\}$, $X_3=\{c,d,e\}$ and $X_4=X \setminus \{a,b,c,d,e\}$ is not covered (note that as $n \geq 6$, $X_4$ is not empty). Therefore, there must be an additional quadruple duplicating at least one more of the already present triples. This leads to $k\geq \frac{1}{4}\ \binom{n}{3}+2$.

Using the fact that $k$ needs to be an integer adds the required ceiling function, so that we get $k \geq \left\lceil\frac{1}{4}\ \binom{n}{3} +2 \right\rceil = \left\lceil\frac{1}{4}\ \binom{n}{3} \right\rceil+2.$
This shows that the first inequality is true for any number $k$ of quadruples in a phylogenetically decisive collection of taxon sets, so it must in particular be true for $k=\widehat{k}(n)$. This completes the proof of the first inequality.

For the second inequality, we prove a slightly different statement which will lead to the desired result: Exploiting the properties of fixing taxa, we inductively prove that Algorithm \ref{alg_phylodecupperboundforminset} returns a phylogenetically decisive set $\sS$ of taxon sets for all $n\geq 6$, and we will subsequently show that this set $\sS$ contains $\binom{n-1}{3}-1$ quadruples.

{\small
\begin{algorithm}[H]
\caption{Construction of a phylogenetically decisive set $\sS$ of taxon sets (subsets of $X=\{1,\ldots,n\}$) with $n\geq 6$ taxa and  with $\binom{n-1}{3}-1$ quadruples}\label{alg_phylodecupperboundforminset}
\LinesNumbered
 \SetKwInOut{Input}{Input}\SetKwInOut{Output}{Output}
 \vspace{0.15 cm}
 \Input{$n$: $n\in \mathbb{N}_{\geq 6}$ }
 \Output{phylogenetically decisive set $\sS$ consisting of $\binom{n-1}{3}-1$ quadruples}
 \Init{}{
 $\sS \gets \left\{\{1,2,3,5\},\{1,2,4,5\},\{1,2,4,6\},\right.$  $\left.\{1,3,4,6\}, \{1,3,5,6\}, \{1,4,5,6\}, \{2,3,4,5\}, \{2,3,4,6\}, \{2,3,5,6\}\right\}$\\ 
 $counter \gets 6$
 }

\While{$ counter < n$}{

$counter \gets counter +1$
\\
Append to $\sS$ all $\binom{counter-2}{2}$ quadruples containing the pair $(1,counter)$ 
}

\Return{$\sS$}
\end{algorithm}
}

For $n=6$, the algorithm does not do anything else than return the result provided by the initialization step, which is the set $\sS$ we have already seen in Example \ref{badexample}. Here, we have $X=\{1,2,3,4,5,6\}$ and $\sS=\left\{\{1,2,3,5\},\{1,2,4,5\},\right.$  $\left.\{1,2,4,6\},\{1,3,4,6\}, \{1,3,5,6\}, \{1,4,5,6\}, \{2,3,4,5\}\right.$,\\ $\left. \{2,3,4,6\}, \{2,3,5,6\}\right\}$, and we have already seen that $\sS$ is phylogenetically decisive (cf. Table \ref{tab_jannesEx} in the Appendix). As $\sS$ contains 9 quadruples and as $\binom{6-1}{3}-1=9$, this proves the base case of the induction.

Next, assume the algorithm works for all natural numbers up to $n-1$. We now consider $n$. The algorithm takes the (by the inductive hypothesis: phylogenetically decisive) set constructed for $n-1$, say $\sS'$, increases the counter from $n-1$ to $n$ and then adds the $\binom{counter-2}{2}=\binom{n-2}{2}$ quadruples that contain the pair $(1,n)$ to form $\sS$. Note that while $\sS'$ does not contain any quadruples that contain taxon $n$, it is phylogenetically decisive for $X'=\{1,\ldots,n-1\}$ by the inductive hypothesis. So the only quadruples that can potentially not get resolved, i.e., that might \emph{not} be uniquely resolved by all possible supertrees, all must contain taxon $n$ (and cannot contain taxon 1). So they must be of the form $\{a,b,c,n\}$, where $a,b,c \in \{2,\ldots,n-1\}$. However, as we have \emph{all} possible quadruples that contain both taxon $n$ and taxon 1, taxon 1 acts as a fixing taxon for all quadruples containing taxon $n$. This is the high-level idea of the proof, and we now elaborate on the details. 

Consider the quadruple $\{a,b,c,n\}$, where $a,b,c \in \{2,\ldots,n-1\}$. As we have the quadruples $\{1,a,b,n\}$, $\{1,a,c,n\}$, $\{1,b,c,n\}$ in $\sS$ (all of them were added in the last \textsf{while}-loop of Algorithm \ref{alg_phylodecupperboundforminset}), and as we know that the quadruple $\{1,a,b,c\}$ gets uniquely resolved by all possible supertrees (because this quadruple was already either present in $\sS'$ or resolved by the quadruples in this set as $\sS'$ is phylogenetically decisive by the inductive hypothesis), taxon 1 is a fixing taxon for the quadruple $\{a,b,c,n\}$. By Proposition \ref{fixpointresolution}, this implies that the quadruple $\{a,b,c,n\}$ gets resolved in a unique way, too. In particular, this shows that $\sS$ cannot induce any CQs that cannot get resolved, which proves the phylogenetic decisiveness of $\sS$ as produced by Algorithm \ref{alg_phylodecupperboundforminset}. 

It only remains to show that $\sS$ has the desired size. In this regard, note that as in each step $\binom{counter-2}{2}$ many quadruples are added for all values of $counter =7,\ldots, n$, clearly the set $\sS$ returned by the algorithm contains $9+ \sum\limits_{i=7}^n \binom{i-2}{2}$ many quadruples. We now rearrange this term as follows:

\begin{align*}
   9+ \sum\limits_{i=7}^n \binom{i-2}{2} &= 9+\underbrace{\sum\limits_{i=2}^{n} \binom{i-2}{2}}_{\overset{(*)}{=}\binom{n-1}{3}}- \underbrace{\sum\limits_{i=2}^6 \binom{i-2}{2}}_{=10} 
   = \binom{n-1}{3}-1.
\end{align*}

Here, equation (*) holds due to the well-known identity $\binom{n+2}{3}=\sum\limits_{i=1}^n \binom{i+1}{2}$ given by the tetrahedral numbers. 

So clearly, the set $\sS$ generated by Algorithm \ref{alg_phylodecupperboundforminset} contains the desired number of quadruples and is phylogenetically decisive. This completes the proof.
\end{proof}

We end this section by pointing out that the lower bound $\left\lceil \frac{1}{4} \binom{n}{3}\right\rceil +2$ for $k$ is not generally tight. For instance, for $n=6$, the bound is  $\left\lceil \frac{1}{4} \binom{6}{3}\right\rceil +2 =7$, whereas we could confirm by exhaustive search that $\widehat{k}(6)=9$ (cf. Table \ref{tab:boundslist}), which equals the upper bound. 

\begin{remark} Note that the lower bound induced by Theorem \ref{thm_lowerboundphylodec} is already a huge improvement to another natural bound. It can easily be seen that $\widehat{k}(n) \geq \left\lceil\frac{S_2(n,4)}{4^{n-4}}\right\rceil$, where $S_2(n,4)$ denotes the Stirling number of the second kind of $n$ and $4$, which in turn denotes the number of partitions of a set of size $n$ into four non-empty and non-overlapping subsets. This inequality is due to the fact that each quadruple $\{a,b,c,d\}$ in a collection of taxon sets $S$ on $X=\{1,\ldots,n\}$ covers precisely $4^{n-4}$ partitions: Starting with separating $a$, $b$, $c$ and $d$ into separate subsets of $X$, all other $n-4$ taxa can choose freely amongst the 4 subsets. If $S$ then contains $k$ quadruples, $S$ covers at most $k \cdot 4^{n-4}$ partitions (note that the number might be smaller as some quadruples in $S$ might cover the same partition). But in order for $S$ to be phylogenetically decisive, by Theorem \ref{sanderson_steel}, \emph{all} partitions of $X$ into 4 non-empty and non-overlapping subsets must be covered, and there are $S_2(n,4)$ many such partitions. So we know that $k \cdot 4^{n-4} \geq S_2(n,4)$, which implies $k \geq \frac{S_2(n,4)}{4^{n-4}}$. Using $k \in \mathbb{N}$ then allows us to add the required ceiling function. 
However, using the well-known identity $S_2(n,4)= \frac{1}{4!}\sum\limits_{j=0}^4 (-1)^{4-j}\binom{4}{j}j^n$, one can easily show that the rather intuitive lower bound based on the Stirling numbers of the second kind is unfortunately generally smaller than the one stated in Theorem \ref{thm_lowerboundphylodec}, which is why this does not improve the bound further.
\end{remark}

Note that the upper bound given by Theorem \ref{thm_lowerboundphylodec} is not generally tight, either. While it is tight for $n=6$, it is not for $n=7$: In our random search we have found sets of 17 quadruples which are phylogenetically decisive (cf. Section \ref{sec_sim}), whereas the upper bound suggests a value of $\binom{7-1}{3}-1=19$ in this case.

\subsection{Differences between phylogenetic decisiveness and fixing taxon traceability}\label{sec_diff}

As we have seen in the previous sections, fixing taxon tracability and phylogenetic decisiveness are closely related. In fact, by Theorem \ref{thm_FTTmain}, we know that every set $\sS$ of taxon sets that is fixing taxon traceable is also phylogenetically decisive. However, Example \ref{badexample} has shown that there are collections of taxon sets that are phylogenetically decisive, but not fixing taxon traceable. In fact, Algorithm \ref{alg_phylodecupperboundforminset} systematically  constructs such sets for all $n\geq 6$ as by Corollary \ref{cor_lowerbound}, the constructed sets contain too few quadruples to be fixing taxon traceable, but by the proof of Theorem \ref{thm_lowerboundphylodec} they are phylogenetically decisive. In this section, we thus want to elaborate on how different the two concepts of fixing taxon traceability and phylogenetic decisiveness really are.

\subsubsection{\texorpdfstring{Theoretical considerations on the number of quadruples in collections of taxon sets that are phylogenetically decisive but not fixing taxon traceable}{Theoretical considerations on the number of quadruples in collections of taxon sets that are phylogenetically decisive but not fixing taxon traceable}}

In this section we derive some bounds on the number of quadruples a set $\sS$ of taxon sets can have in order to make the concepts of phylogenetic decisiveness and fixing taxon traceability either differ or coincide. We start with the following remark.

\begin{remark} If the number of taxa $n$ equals 4 or 5, the concepts of fixing taxon traceability and phylogenetic decisiveness coincide. In particular, if $n=4$, then there is only one quadruple. This quadruple has to be contained in $\sS$ in order to make $\sS$ fixing taxon traceable and phylogenetically decisive. If this quadruple is not present, i.e., if $\sS=\emptyset$, then $\sS$ is neither fixing taxon traceable nor phylogenetically decisive. Similarly, for $n=5$, it can easily be seen that any four out of the $\binom{5}{4}=5$ possible quadruples need to be present in $\sS$ for $\sS$ to be phylogenetically decisive and fixing taxon traceable. With fewer than four quadruples, $\sS$ does not fulfull either one of these properties.
\end{remark}

Next, we state the following corollary on cases when the two concepts of fixing taxon traceability and phylogenetic decisiveness are guaranteed to coincide.

\begin{corollary}\label{cor_sizeskforwhichFTT=PhyloDec}
Let $X=\{1,\ldots,n\}$ be a collection of taxon sets and let $\sS$ be a set of subsets of $X$ containing $k$ quadruples. Then, if $k< \left\lceil \frac{1}{4}\binom{n}{3}  \right\rceil+2$, $\sS$ is neither phylogenetically decisive nor fixing taxon traceable. Moreover, if $k\geq\binom{n}{4}-n+4$, then $\sS$ is both phylogenetically decisive and fixing taxon traceable. In particular, in both cases, fixing taxon traceability and phylogenetic decisiveness are equivalent. 
\end{corollary}

\begin{proof} If $k< \left\lceil \frac{1}{4}\binom{n}{3}  \right\rceil+2$, by Theorem \ref{thm_lowerboundphylodec}, $\sS$ is \emph{not} phylogenetically decisive. By Theorem \ref{thm_FTTmain}, this implies that $\sS$ is also not fixing taxon traceable. This shows that when $k< \left\lceil \frac{1}{4}\binom{n}{3}  \right\rceil+2$, the two concepts coincide. 

Now, if $k\geq\binom{n}{4}-n+4$, by Theorem \ref{thm_upperbound}, $\sS$ is fixing taxon traceable and thus, again by Theorem \ref{thm_FTTmain}, $\sS$ is also phylogenetically decisive. This shows that when $k\geq\binom{n}{4}-n+4$, the two concepts coincide, which completes the proof.
\end{proof}

 Note that Corollary \ref{cor_sizeskforwhichFTT=PhyloDec} is a generalization of Theorem \ref{thm_moan}, and the corollary's proof, by using the new concept of fixing taxon traceability, gives an alternative proof to said theorem. However,  Corollary \ref{cor_sizeskforwhichFTT=PhyloDec} shows that for the concepts of fixing taxon traceability and phylogenetic decisiveness to differ, we need the number $k$ of quadruples in a set $\sS$ of taxon sets to be contained in the interval $\left[\left\lceil \frac{1}{4}\binom{n}{3}  \right\rceil+2,\binom{n}{4}-n+3\right]$. However, it is the aim of the remainder of this section to analyze the upper bound of this interval a bit more in-depth. In fact, we believe that this upper bound is not tight, cf. Conjecture \ref{con_upperbound} at the end of this section, i.e., we think that the number of quadruples in \emph{all} collections of taxon sets that are phylogenetically decisive without being fixing taxon traceable is smaller than suggested by this bound.

But how large can $k$ be for $\sS$ to be phylogenetically decisive without also being fixing taxon traceable? The following theorem states a lower bound for the maximum value of $k$.

\begin{theorem}\label{thm_upperboundFTTneqPhyloDec} Let $n \in \mathbb{N}_{\geq 6}$ and let $X=\{1,\ldots,n\}$ be a collection of taxon sets. Then, there is a set $\sS$ of subsets of $X$ containing $k= \binom{n}{4}-3n+13$ quadruples such that $\sS$ is phylogenetically decisive, but not fixing taxon traceable, and that $\sS$ is maximal with these properties. In particular, no quadruple can be added to $\sS$ without making it fixing taxon traceable.
\end{theorem}

\begin{proof}
We prove the assertion by proving the correctness of Algorithm \ref{alg_phylodecNonFTT} by induction on $n$. For $n=6$, consider the set $\sS$ with $10=\binom{6}{4}-5=\binom{6}{4}-3\cdot 6+13$ quadruples provided by the initialization step of Algorithm \ref{alg_phylodecNonFTT}. By exhaustively considering all 65 partitions of $X=\{1,\ldots,6\}$, one can easily verify that $\sS$ is phylogenetically decisive; cf. Table \ref{tab_maxDecNonFTT} in the Appendix.

{\small
\begin{algorithm}[H]
\caption{Construction of a phylogenetically decisive, but non-FTT set $\sS$ of taxon sets on $n\geq 6$ taxa with $\binom{n}{4}-3n+13$ quadruples}\label{alg_phylodecNonFTT}
\LinesNumbered
 \SetKwInOut{Input}{Input}\SetKwInOut{Output}{Output}
 \vspace{0.15 cm}
 \Input{$n$: $n\in \mathbb{N}_{\geq 6}$ }
 \Output{phylogenetically decisive and non-fixing taxon traceable set $\sS$ consisting of $\binom{n}{4}-3n+13$ quadruples}
 \Init{}{
 $\sS \gets \binom{X}{4} \setminus  \left\{\{1,2,5,6\},\{1,3,4,6\},\{1,4,5,6\},\{2,3,4,6\},\{2,3,5,6\}\right\}$\\ 
 $counter \gets 6$
 \\ 
 $newGrays\gets \emptyset$ \\
  $newWhites\gets \emptyset$ \\
  $X_c \gets \emptyset$\\
  $\widetilde{X}\gets\emptyset$
 }

\While{$n-counter\geq 1$}{
$counter \gets counter +1$\\
$X_{c} \gets \{1,\ldots,counter-1\} $ \\
$\widetilde{X} \gets \{ \{a,b,c,counter\} \ | \ a,b,c \in X_{c}\}$ \\
$newGrays\gets \{\{2,5,6,counter\},\{3,4,6,counter\},\{4,5,6,counter\}\}$\\
$newWhites \gets \widetilde{X} \setminus newGrays$\\
$\sS \gets \sS \cup newWhites$  
}

\Return{$\sS$}
\end{algorithm}
}

Moreover, $\sS$ is not fixing taxon traceable. In order to see that, consider the 3-overlap-graph $K_4^6$ depicted in Figure \ref{fig_phylodecNonFTT_basecase}. As the cross quadruples form a cycle in that graph, this shows they all have two CQ neighbors. Thus, for each cross quadruple, none of the two possible candidates for fixing taxa actually works. So $\sS$ is not fixing taxon traceable. This completes the base case of the induction.

 \begin{figure}[ht]      \centering\vspace{0.5cm} 
    \includegraphics[width=10cm]{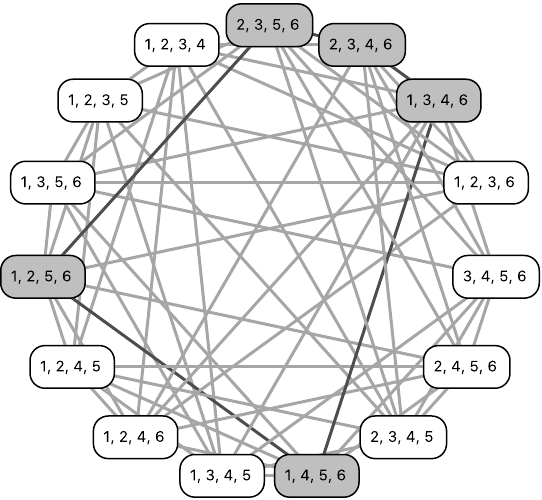} 
    \caption{The 3-overlap graph $K_4^6$ as induced by the initialization phase of Algorithm \ref{alg_phylodecNonFTT}: The cross quadruples form a cycle. }\label{fig_phylodecNonFTT_basecase}
  \end{figure}

Now for the inductive step, assume the algorithm is already known to produce a phylogenetically decisive set $\sS$ of taxon sets that is not fixing taxon traceable for up to $n-1$ taxa, and now consider $X=\{1,\ldots,n\}$. In the \textsf{while}-loop, the algorithm takes the set produced for $n-1$ and adds all but three possible quadruples containing taxon $n$ to it. The three quadruples not added to $\sS$ are $\{2,5,6,n\}$, $\{3,4,6,n\}$ and $\{4,5,6,n\}$. 

We first verify that $\sS$ produced as described has the correct size. Thus, note that there are $\binom{n-1}{3}$ many quadruples containing taxon $n$, but three of them do not get added. In summary, the number $k$ of quadruples in $\sS$ can be calculated as follows:

\begin{align*}k&= \underbrace{\binom{n-1}{4}-3(n-1)+13}_{\mbox{by induction}} + \underbrace{\binom{n-1}{3}-3}_{\mbox{added in \textsf{while}-loop}}\\
&= \binom{n}{4}-3n+13,\end{align*}
which shows that $\sS$ has the claimed size.

Now we need to show that $\sS$ also has the proclaimed properties. We start by showing that $\sS$ is phylogenetically decisive. As the set was decisive for $n-1$ before the new quadruples got added in the \textsf{while}-loop, all 4-partitions of $X=\{1,\ldots,n\}$ in which taxon $n$ is \emph{not} forming a set of its own need \emph{not} be considered (they are already covered by quadruples not containing taxon $n$). So we only need to consider partitions of the kind $X_1|X_2|X_3| \{n\}$. However, these are necessarily all covered, as taxon 1 is contained in one of the sets $X_i$, $i \in \{1,2,3\}$, and \emph{all} quadruples containing both taxon 1 and taxon $n$ are contained in $\sS$, as $\{2,5,6,n\}$, $\{3,4,6,n\}$ and $\{4,5,6,n\}$ are the only quadruples containing $n$ that were not added to $\sS$, but these do not contain taxon 1. So if taxon 1 is contained in, say, $X_1$, we choose $x_2 \in X_2$ and $x_3 \in X_3$ arbitrarily and consider the quadruple $\{1,x_2,x_3,n\} \in \sS$. This covers the partition in question as required, which shows that $\sS$ is phylogenetically decisive.

Next, we need to show that $\sS$ is not fixing taxon traceable. However, this can be easily seen by considering the construction of $\sS$: In each step, it is made sure that the newly added taxon does not act as a fixing taxon for any previous cross quadruple. This is due to the way the new cross quadruples are chosen in the \textsf{while}-loop. For instance, $\{2,5,6,n\}$ makes sure that $n$ cannot act as a fixing taxon for $\{1,2,5,6\}$, $\{2,3,5,6\}$ or $\{2,5,6,i\}$ with $i \in \{7,\ldots,n-1\}$. Analogously, $\{3,4,6,n\}$ makes sure that $n$ is not a fixing taxon for $\{1,3,4,6\}$, $\{2,3,4,6\}$ or $\{3,4,6,i\}$ with $i \in \{7,\ldots,n-1\}$. Last but not least, $\{4,5,6,n\}$ makes sure that $n$ is not a fixing taxon for $\{1,4,5,6\}$ or $\{4,5,6,i\}$ with $i \in \{7,\ldots,n-1\}$. This shows that Algorithm \ref{alg_phylodecNonFTT} starts with a non-fixing taxon traceable set in the initialization step and in each iteration of the \textsf{while}-loop makes sure that the original cross quadruples cannot get resolved. This shows that $\sS$ is not fixing taxon traceable. 

It remains to show that $\sS$ is maximal with the desired properties, i.e., that adding any additional quadruple to $\sS$ would lead to a set that is either not phylogenetically decisive or fixing taxon traceable. The first cannot happen (a set cannot lose decisiveness by adding more input information in the form of more quadruples), so it remains to show that adding any cross quadruple to $\sS$ would make it fixing taxon traceable. 

For $n=6$, consider again the 3-overlap-graph depicted in Figure \ref{fig_phylodecNonFTT_basecase}. Re-coloring one CQ in this graph white (by adding the corresponding cross quadruple to $\sS$) causes two remaining CQs to have only one CQ neighbor. These two CQs thus both have a fixing taxon. For instance, if $\{1,2,5,6\}$ gets added to $\sS$, $\{2,3,5,6\}$ has fixing taxon 1 and $\{1,4,5,6\}$ has fixing taxon 2. The others follow analogously. Moreover, it can easily be seen that then, all other quadruples iteratively also get resolved, which is why $\sS$ is fixing taxon traceable. So for $n=6$, maximality is shown.

For $n>6$, it is again clear that none of the original five cross quadruples can be added, because by the same argument they would then all five be resolved, so that each  one of the new cross quadruples $\{2,5,6,n\}$, $\{3,4,6,n\}$ and $\{4,5,6,n\}$ would have a fixing taxon in $\{1,\ldots,6\}$. 

However, if one of these new cross quadruples were added to $\sS$, it would make $n$ a fixing taxon for all of the original five cross quadruples from the initialization step with which it shares three taxa. This, in turn, would by the same arguments as above resolve all of these original cross quadruples, so that again all new cross quadruples would have a fixing taxon in  $\{1,\ldots,6\}$.

This shows that no cross quadruple can be added to $\sS$ without making $\sS$ fixing taxon traceable. This shows that $\sS$ is maximal and thus completes the proof.

\end{proof}

\begin{remark} The simulation results we  present in Section \ref{sec_sim} suggest that the bound of $k=\binom{n}{4}-3n+13$ implied by Theorem \ref{thm_upperboundFTTneqPhyloDec} is actually sharp. This leads to the following conjecture, which we leave for future research.\end{remark}

\begin{conjecture}\label{con_upperbound} Let $X=\{1,\ldots,n\}$ with $n \in \mathbb{N}_{\geq 6}$. Let $\sS$ be a set of subsets of $X$ inducing $k$ quadruples, such that $\sS$ is phylogenetically decisive, but not fixing taxon traceable. Then, $k \leq \binom{n}{4}-3n+13$.
\end{conjecture}

\subsubsection{Simulation approach}\label{sec_sim} 
In this subsection, we complement all theoretical considerations of the previous subsection with simulations highlighting cases in which the two concepts of phylogenetic decisiveness and fixing taxon traceability coincide or differ. More precisely, we analyze how powerful Algorithm \ref{alg_generalFTT}, which determines if a given collection of taxon sets is fixing taxon traceable, is concerning the detection of phylogenetic decisiveness: While we know from Theorem \ref{thm_FTTmain} that all fixing taxon traceable collections of taxon sets are also phylogenetically decisive, we also  already know from Theorem \ref{thm_upperboundFTTneqPhyloDec} that not converse is not true, so Algorithm \ref{alg_generalFTT} will miss some sets if used to detect decisiveness. It is the aim of this section to analyze how often this happens with randomly sampled sets of quadruples. This is of high interest as answering the question if a given collection of taxon sets is phylogenetically decisive is coNP-complete.  Algorithm \ref{alg_generalFTT} can be used as a heuristical approach to tackle this problem. In our simulation, we analyze the power of this approach.

All data was generated within \textsf{R} (v4.2)\cite{RCoreTeam2019}, and both the R package \verb+FixingTaxonTraceR+ that includes the \textsf{R} implementation of \textsf{Algorithm 1} and the simulation code are available on \textsf{github} \cite{githubsoftware,githubsimulation}. In the following, we describe our simulation approach within the ADEMP framework established in \cite{morris2019}. 

\textbf{Simulation setup}

\textbf{(A)ims}: The aim of the simulation was to quantify the number of collections of taxon sets for which the concepts of phylogenetic decisiveness and fixing taxon traceability coincide or differ, respectively. In more detail, we aimed at (1) estimating the power of  Algorithm \ref{alg_generalFTT} as a formal test for phylogenetic decisiveness and (2) identifying bounds in which Algorithm \ref{alg_generalFTT}  will always differ from decisiveness as described in the sections above.

\textbf{(D)ata-generating mechanisms}: We performed simulations for $n=6, \ldots, 10$ taxa. For each value of $n$, we considered taxon set $X=\{1,\ldots,n\}$ and randomly sampled $k_n$ quadruples from $X$ for all $k_n= \left\lceil\frac{1}{4}\ \binom{n}{3} \right\rceil +2,\ldots, \binom{n}{4} - n + 3$. Note that the possible values of $k_n$ were chosen according to the known bounds for fixing taxon traceability (cf. Theorem \ref{thm_upperbound} and Theorem \ref{thm_lowerboundphylodec}).

So for each $k_n$, we sampled $\sS_i=\{Y_1, \ldots, Y_{k_n}\}$ as random subsets of quadruples of $X$, with $i=1, \ldots, 10,000=:n_{sim}$; i.e., we repeated each such sampling 10,000 times. The quadruples were drawn uniformly at random out of all possible sets containing $k_n$ quadruples using the \textsf{R}-function  \textsf{sample()}. However, for $n=6$ and each possible value of $k_n$, there are in total fewer than 10,000 possible combinations.\footnote{To see this, note that there are $\binom{6}{4}=15$ quadruples, and we considered $k_6\in\{\left\lceil\frac{1}{4}\ \binom{6}{3} \right\rceil +2=7, ..., \binom{6}{4} - 6 + 3=12\}$. As $\binom{15}{7}=\binom{15}{8}=6,435$ and as $\binom{15}{k_6}<6,435$ if $k_6>8$, we could not draw 10,000 samples for any possible value of $k_6$. Instead, we performed exhaustive searches for all values of $k_6$.} Hence, in this case we used all possible combinations of quadruples precisely once (exhaustive search).

\textbf{(E)stimands or targets}: The target of interest was the following null hypothesis: 

$H_0$: The given set $\sS_i$ of quadruples is not phylogenetically  decisive. 

The alternative hypothesis was: 

$H_1$: The given set $\sS_i$ of quadruples is phylogenetically decisive. 

\textbf{(M)ethods}: For each simulated set of quadruples $\sS_i$, we used Algorithm \ref{alg_generalFTT} to decide if $\sS_i$ is fixing taxon traceable. To check for phylogenetic decisiveness, the four-way partition property was tested, i.e., we exhaustively checked all possible 4-partitions of $X=\{1,\ldots,n\}$ for coverage by $\sS_i$. The two results for the set $\sS_i$ were returned in a logical vector $x_i$, in which the first entry refers to fixing taxon traceability and the second one to phylogenetic decisiveness. In particular, (\textsf{T},\textsf{T}) indicates both fixing taxon traceability and phylogenetic decisiveness (denoted as true positive, as both the actual condition and test are positive), (\textsf{F},\textsf{T}) indicates only phylogenetic decisiveness (denoted as false negative, as the actual condition is positive but the test is negative), and (\textsf{F},\textsf{F}) indicates no phylogenetic decisiveness (denoted as true negative, as both the actual condition and the test are negative).\footnote{Note that in our setting (\textsf{T},\textsf{F}), i.e., fixing taxon traceability without phylogenetic decisiveness (false positive), is not possibly due to Theorem \ref{thm_FTTmain}.} The result of each simulated set was stored in the matrix $M \in \mathbb{R}^{10,000 \times 2}$, which was then used to create a contingency table for each value of $k_n$. 

\textbf{(P)erformance measures}: To assess the performance of 
Algorithm \ref{alg_generalFTT}, which was developed to detect fixing taxon traceability, as an indicator for phylogenetic decisiveness, we derived three parameters of contingency tables for all $k_n$: 

\begin{itemize} 
    \item True Positive Rate (TPR) as the proportion of fixing taxon traceable sets (true positive, TP) to phylogenetically decisive sets (actual condition positive, P): $TPR = \frac{TP}{P}$. This value is also known as \textbf{power} of a test.
    \item Negative Predictive Value (NPV) as the proportion of phylogenetically non-decisive sets (true negative, TN) to sets that are not fixing taxon traceable (predicted negative, either true or false): $NPV = \frac{TN}{TN+FN}$.
\end{itemize}

Additionally, for all analyzed values of $n$ and $k_n$, we calculated a third parameter, namely the prevalence (PREV): PREV describes the proportion of phylogenetically decisive sets (positives, $P$) in our random sample sets: $PREV = \frac{P}{P+N}$, with $N$ being all non-decisive sets. Note that the prevalence was tested using the four-way partition property, not Algorithm \ref{alg_generalFTT}. Thus, PREV contains confirmed information concerning the simulated data and is independent of estimations.

As a summary, for each $n$ we report the medians and interquartile ranges (IQR) of all three parameters over all values $k_n$. In addition, we extracted the following bounds for $k_n$: 

\begin{itemize}
    \item $\min(k_{n})$ for phylogenetic decisiveness: This value describes the minimum number $k_n$ of input quadruples for which phylogenetic decisiveness was observed in at least one case.
    \item $\min(k_{n})$ for fixing taxon traceability: This value describes the minimum number $k_n$ of input quadruples for which fixing taxon traceability was observed in at least one case.
    \item $\max(k_{n})$ for non-decisiveness: This value describes the maximum number $k_n$ of input quadruples for which phylogenetic decisiveness was \emph{not} observed in at least one case. 
    \item $\max(k_{n})$ for phylogenetic decisiveness without fixing taxon traceability: This value describes the maximum number $k_n$ of input quadruples for which phylogenetic decisiveness was observed without fixing taxon traceability in at least one case. 
    \item $\max(k'_{n})$ for decisiveness without traceability, defined as the sum of input quadruples and quadruples resolved by the algorithm: This value refers to the maximum number $k_n'$ of quadruples contained in the input set or the set resolved when Algorithm \ref{alg_generalFTT} stops and gives the largest such value we observed when we were in a case with phylogenetic decisiveness without fixing taxon traceability.\footnote{Note that even in cases in which Algrithm \ref{alg_generalFTT} fails to resolve all quadruples (as the set of input quadruples is not FTT), it often is able to resolve at least some quadruples additional to the ones in the input set. This total number $k_n'$ of quadruples resolved by the algorithm is what the parameter $\max(k'_{n})$ is referring to.
    
    }
\end{itemize}

\textbf{Simulation results}

A summary of the simulation performance is given in Table \ref{tab:ADEMPresults}, where we present the median and the interquartile ranges (IQR) of PREV, TPR and NPV over all tested values of $k_n$.

\begin{table}[ht]
    \centering
    \begin{tabular}{|c||c|c|c|}
    \hline
  $\bm{n}$ &  \textbf{TPR [IQR]} & \textbf{NPV [IQR]} & \textbf{PREV [IQR]} \\
    \hline\hline
    \bf{6} & 0.97 [0.71,1] & 1 [0.97,1] & 0.62 [0.36,0.83] \\ 
    \bf{7} &  0.99 [0.65,1] & 1 [0.98,1] & 0.74 [0.25,0.96]  \\
    \bf{8} & 1.00 [0.80,1] & 1 [0.98,1] & 0.85 [0.29,0.99] \\
    \bf{9} &  1.00 [0.91,1] & 1 [0.99,1] & 0.91 [0.32,1] \\
    \bf{10} &  1.00 [0.97,1] & 1 [1,1] & 0.95 [0.44,1] \\    
    \hline
    \end{tabular}
    \caption{Summary of simulation performance measures: median and interquantile range of the True Positive Rate (TPR), Negative Predictive Value (NPV) and Prevalence (Prev) are given. These summary measures were all obtained using $k_n \geq \min(k_{n})$ for phylogenetic decisiveness.}
    \label{tab:ADEMPresults}
\end{table}

For all analyzed values of $n$, we estimated a high power of Algorithm \ref{alg_generalFTT} to identify phylogenetic decisiveness. The power (as given by TPR) increases with $k_n$, as more quadruples increase the chance of finding a fixing taxon, cf. Figure \ref{fig_simresults}. NPV also increases with increasing numbers of input quadruples, which is another indicator that Algorithm \ref{alg_generalFTT} works well when sufficiently many input quadruples are given: It indicates that the false negatives become neglible as $k_n$ grows. 

PREV, on the other hand, also tends to 1 as $k_n$ grows. This observation is independent of Algorithm \ref{alg_generalFTT}: It indicates that the more input quadruples we have, the more likely it is that this set of quadruples is phylogenetically decisive. To get a definitive value, this was explicitly tested using the four-way partition property, not Algorithm \ref{alg_generalFTT}.

\begin{figure}
\vspace{0.5cm} 
   \includegraphics[width=12cm]{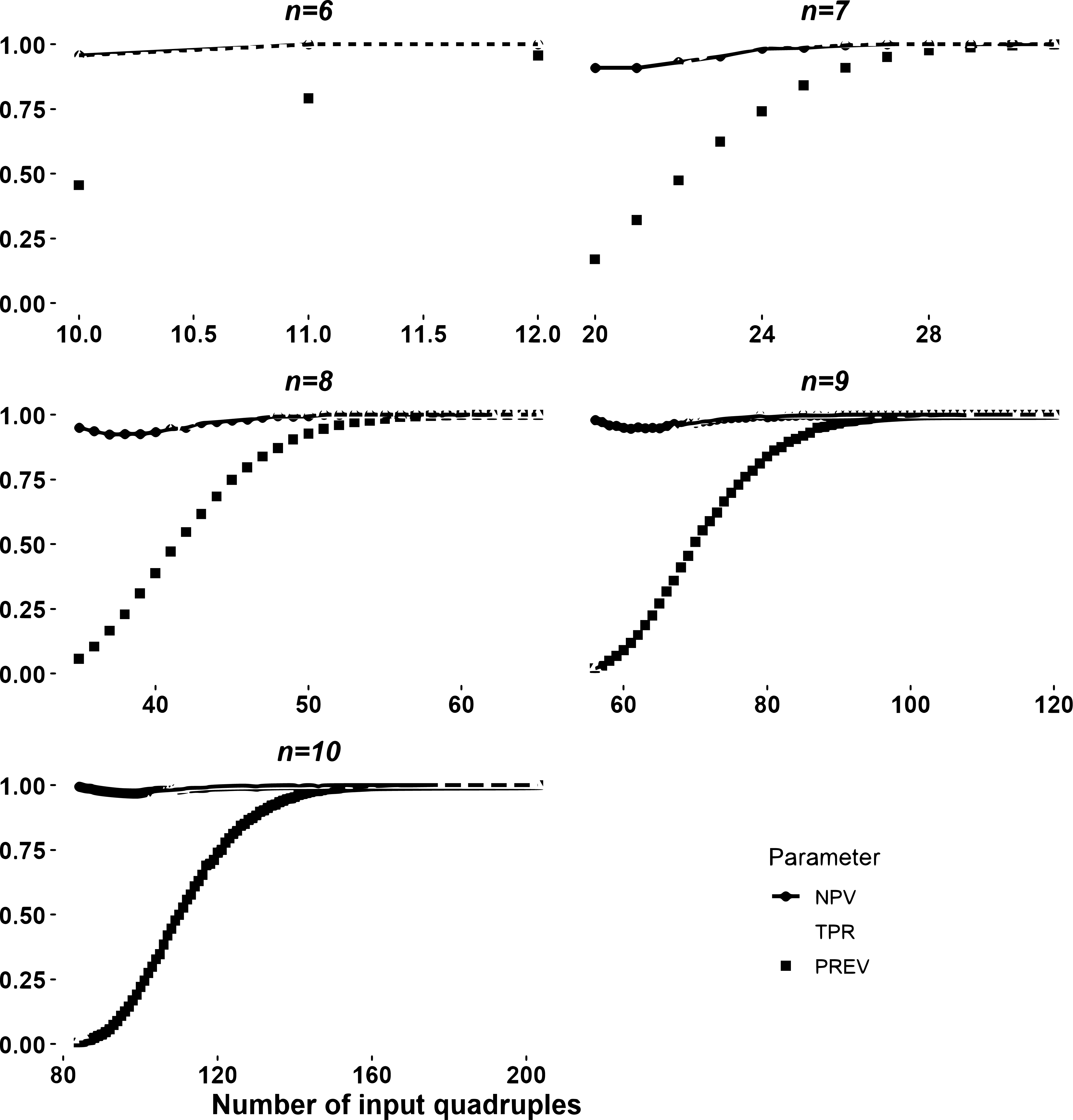} 
    \caption{Overview of the parameters PREV, TPR and NPV for different numbers of input quadruples. TPR describes the power of Algorithm \ref{alg_generalFTT} to detect phylogenetic decisiveness.}\label{fig_simresults}

\end{figure}

In Table \ref{tab:boundslist}, we list the bounds of phylogentic decisiveness as simulated. For $n=6$, the bounds are guaranteed to be true, as all possible combinations were exhaustively tested. For $n=6$ and $n=7$, the simulation even revealed the true upper bound for decisiveness according to Theorem \ref{thm_moan}.

\begin{table}[ht]
    \centering
\begin{tabular}{ccccccccl}
\cline{4-8}
                       &   & \multicolumn{1}{l|}{}  & \multicolumn{5}{c|}{$n$:}               &                     \\ \cline{1-3} \cline{9-9} 
\multicolumn{1}{|l}{} & 4WPP & \multicolumn{1}{l|}{FTT} & 6 & 7 & 8 & 9 & \multicolumn{1}{l|}{10} & \multicolumn{1}{l|}{comment} \\ 
    \hline\hline
    \multicolumn{1}{|c}{$\min(k_{n})$} & $\checkmark$ & -- & \bf{9}& 17    & 30    &  46   &  72    &  \multicolumn{1}{l|}{simulated} \\
     \multicolumn{1}{|c}{}& $\checkmark$ & -- & \bf{9}& 14    & 26    &  46   &  72    & \multicolumn{1}{l|}{best known } \\
    \multicolumn{1}{|c}{} & $\checkmark$ & -- &  5    & 9     & 14    & 21    & 30     & \multicolumn{1}{l|}{theoretical} \\
           \multicolumn{1}{|c}{}       &   &   &       &       &       &       &        & \multicolumn{1}{l|}{(Th. \ref{thm_lowerboundphylodec})} \\ 
    \hline
    \multicolumn{1}{|c}{$\min(k_{n})$} & $\checkmark$ & $\checkmark$ & \bf{10}& \bf{20}& \bf{35}&  \bf{56}&  \bf{84} & \multicolumn{1}{l|}{simulated} \\
    \multicolumn{1}{|c}{} & $\checkmark$ & $\checkmark$ & \bf{10}& \bf{20}& \bf{35}&  \bf{56}&  \bf{84} & \multicolumn{1}{l|}{tight (Cor. \ref{cor_lowerbound})} \\
    \hline \hline
    \multicolumn{1}{|c}{$\max(k_{n})$} & -- & -- & 12    & 31    & 62    & 113   & 183    & \multicolumn{1}{l|}{simulated} \\
    \multicolumn{1}{|c}{} & -- & -- &\bf{12}&\bf{31}&\bf{65}&\bf{120}&\bf{203}& \multicolumn{1}{l|}{tight (Th. \ref{thm_moan})} \\
    \hline 
    \multicolumn{1}{|c}{$\max(k_{n})$} & $\checkmark$ & -- & 10& 26& 50&  89& 146 & \multicolumn{1}{l|}{simulated}   \\
   \multicolumn{1}{|c}{}  & $\checkmark$ & -- & 10& 27& 59&  112& 193 & \multicolumn{1}{l|}{best known} \\
    \multicolumn{1}{|c}{}
                  &   &   &       &       &       &       &        & \multicolumn{1}{l|}{(Th. \ref{thm_upperboundFTTneqPhyloDec})}\\ \hdashline
     \multicolumn{1} {|c}{$\max(k'_{n})$} & $\checkmark$ & -- & 10& 27& 59& 112& 193 & \multicolumn{1}{l|}{simulated} \\
    \hline
    \end{tabular}
    \caption{Bounds on the number $k_n$ of quadruples found in a set $\sS$ of taxon sets fulfilling the four-way partition property (4WPP =$\checkmark$) or not (4WPP = --), or fulfilling fixing taxon traceability (FTT =$\checkmark$) or not (FTT = --). In some cases, we constructed examples that improved the bounds found by the simulations; in this case, best known results may differ from simulated ones. For comparison, we also list theoretical bounds known from previous sections. All bounds known to be tight (e.g., by exhaustive search) are printed in bold.}
    \label{tab:boundslist}
\end{table}

Moreover, some of the bounds stated in Table \ref{tab:boundslist} found by the simulations we performed are already known not to be best possible. For instance, for $n=7$ and $n=8$ we were able to manually construct phylogenetically decisive collections of taxon sets employing 14 and 26 quadruples, respectively (cf. Example \ref{ex_steinerbased}), but the smallest such sets discovered during our simulation had sizes 17 and 30, respectively. This can be attributed to the number of possible sets of quadruples: For $n=7$, there are $\binom{7}{4}=35$ quadruples, and thus there are $2^{35}-1=34,359,738,367$ possible non-empty sets of quadruples. Similarly, for $n=8$, there are $2^{70}-1$ possible non-empty sets of quadruples. Sampling 10,000 such sets each thus only amounts to an extremely small percentage of quadruple sets, which is why it is not surprising that the correct value cannot always be found. On the other hand, concerning the minimum value $k_n$ for which FTT collections of taxon sets were discovered, our simulations for $n=6,\ldots,10$ always found the correct values as given by Corollary \ref{cor_lowerbound}. In fact, the simulations originally pointed us to the pattern and thus inspired us to state and prove Theorem \ref{thm_generaloverlapcoloring}, which is the basis for said corollary.

\begin{example}\label{ex_steinerbased}
$\mbox{}$
\begin{itemize}
\item Let $X=\{1,\ldots,7\}$ and $\sS=\{\{1, 2, 3, 4\}, \{1, 2, 5, 6\}, \{1, 3, 5, 7\}, \{1, 4, 6, 7\}, $ $\left\{3, 4, 5, 6\}, \{2, 4, 5, 7\},\right.\\ \left. \{2, 3, 6, 7\}, \{1, 2, 4, 6\}, \{1, 2, 5, 7\}, \{1, 3, 4, 5\right\},\{1, 3, 6,7\}, \{4, 5, 6, 7\}, \{2, 3, 5, 6\}, \{2, 3, 4, 7\}\}$. Then, $\sS$ is phylogenetically decisive but not fixing taxon traceable (calculations not shown). Moreover, $\sS$ contains only 14 quadruples and is thus smaller than the smallest set (size 17) recovered by our simulation approach.
\item Let $X=\{1,\ldots,8\}$ and $\sS=\{\{1, 2, 5, 6\}, \{1, 2, 7, 8\}, \{1, 3, 5, 7\}, \{1, 3, 6, 8\},$ $  \{1, 4, 5, 8\}, \{1, 4, 6, 7\},$ \\ $ \{3, 4, 7, 8\}, \{3, 4, 5, 6\}, \{2, 4, 6, 8\}, \{2, 4, 5, 7\},\{2, 3, 6, 7\}, \{2, 3, 5, 8\}, \{1, 2, 3, 8\}, \{1, 2, 4, 6\}, \{1, 2, 5, 7\}, $ \\ $\{1, 3, 4, 5\}, \{1, 3, 6, 7\}, \{1, 4, 7, 8\}, \{1, 5, 6, 8\}, \{4, 5, 6, 7\}, \{3, 5, 7, 8\}, \{3, 4, 6, 8\}, \{2, 6, 7, 8\}, \{2, 4, 5, 8\}, $ \\ $ \{2, 3, 5, 6\}, \{2, 3, 4, 7\}\}$. Then, $\sS$ is phylogenetically decisive but not fixing taxon traceable (calculations not shown). Moreover, $\sS$ contains only 26 quadruples and is thus smaller than the smallest set (size 30) recovered by our simulation approach.
\end{itemize}
\end{example}

\section{Discussion and outlook}\label{discussion}
Deriving new information on taxa by combining various compatible trees on different taxon sets into one common supertree is a challenging task, even if the input taxon sets are compatible with one another. It is therefore understandable that phylogeneticists seek to understand a priori if a particular set of input trees has the potential to deliver new information -- in the perfect case even to provide a unique supertree. 

However, the phylogenetic decisiveness problem has recently been shown to be coNP-complete. Therefore, finding classes of collections of taxon sets which are guaranteed to be phylogenetically decisive and which can be identified in polynomial time is a relevant task. One such class is the class of fixing taxon traceable collections of taxon sets, which we introduced in this manuscript. Finding more such classes and investigating the differences between them and phylogenetic decisiveness more in-depth will be an interesting topic for future research. 

In order to investigate these differences, it may help to expand the collections of taxon sets under investigation by adding those CQs that do have fixing taxa, i.e., by running Algorithm  \ref{alg_generalFTT} and to analyze the resulting collection of taxon sets instead, even if it is not fixing taxon traceable. This way, our algorithm actually helps to sort cross quadruples into unproblematic and problematic ones, which can be helpful for biologists. For instance, this might help to decide which taxa to take out of a study as they are contained in too many problematic (i.e., not fixing taxon resolvable) quadruples, or -- on the contrary -- which additional quadruples need to be sampled to guarantee decisiveness.

Moreover, while fixing taxon traceability has allowed us to answer some questions concerning decisiveness, a tight lower bound for the number of quadruples needed for a collection of taxon sets to be phylogenetically decisive is still unknown. Theorem \ref{thm_lowerboundphylodec} gives a range for this number, but as we have seen in our simulations, the lower end of this interval is often pretty far from the number of quadruples that were really needed, which is why we suspect that the upper end of the interval might be closer to the truth. Similarly, a tight upper bound for the number of quadruples in a collection of taxon sets that is phylogenetically decisive but not fixing taxon traceable is also still unknown, but -- again, based on our simulations -- we conjecture that the number of quadruples induced by Theorem \ref{thm_upperboundFTTneqPhyloDec} is a tight bound, cf. Conjecture \ref{con_upperbound}. We leave it open for future research. 

So while the questions posed here are beyond the scope of the current manuscript, we are confident that they will inspire new research ideas in future. In particular, we believe that fixing taxon traceable collections of taxon sets will play an important role in biological applications as they can be identified efficiently and as it has only been proven recently that this is not the case for all phylogenetically decisive collections of taxon sets in general.

\subsection*{Acknowledgement}
We wish to thank Mike Steel and Leo van Iersel for helpful discussions on the topic, including but not limited to insights concerning the relationship between Problems \ref{decprob} and \ref{rainbowprob}. Moreover, we thank Steven Kelk for patiently answering some complexity questions. 

\subsection*{Statements and Declarations}

\subsubsection*{Data availability}
The authors declare that the data supporting the findings of this study are available in the Github repository accompanying this manuscript \cite{githubsimulation}.

\subsubsection*{Funding}
The authors declare that no funds, grants, or other support were received during the preparation of this manuscript.

\subsubsection*{Competing Interests}
The authors have no relevant financial or non-financial interests to disclose.

\subsubsection*{Author contributions}
MF contributed the mathematical proofs presented in the manuscript. JP contributed Example \ref{badexample} as well as the simulation study presented in Section \ref{sec_sim}. JP also provided the \textsf{R} software and package \verb+FixingTaxonTraceR+ accompanying this manuscript. Both authors contributed to the study conception and design. The first draft of the manuscript was written by MF, and both authors commented on previous versions of the manuscript. All authors read and approved the final manuscript.

\bibliographystyle{plainnat}      
 \bibliography{fischer_bibfile}

 \newpage 
 \section*{Appendix}

\begin{table}\tiny
$
\begin{tabular}{cccc||c}
\multicolumn{4}{c||}{input trees}& \text{possible supertrees} \\ \hline 
 \text{(1, 2, (3, 4))} & \text{(1, 2, (3, 5))} & \text{(1, 3, (4, 5))} & \text{(2, 3, (4, 5))} & \text{(1, 2, (3, (4, 5)))} \\
 \text{(1, 3, (2, 4))} & \text{(1, 2, (3, 5))} & \text{(1, 3, (4, 5))} & \text{(2, 3, (4, 5))} &   \\
 \text{(1, 4, (2, 3))} & \text{(1, 2, (3, 5))} & \text{(1, 3, (4, 5))} & \text{(2, 3, (4, 5))} &   \\
 \text{(1, 2, (3, 4))} & \text{(1, 3, (2, 5))} & \text{(1, 3, (4, 5))} & \text{(2, 3, (4, 5))} &   \\
 \text{(1, 3, (2, 4))} & \text{(1, 3, (2, 5))} & \text{(1, 3, (4, 5))} & \text{(2, 3, (4, 5))} & \text{(1, (2, (4, 5)), 3)} \\
 \text{(1, 4, (2, 3))} & \text{(1, 3, (2, 5))} & \text{(1, 3, (4, 5))} & \text{(2, 3, (4, 5))} &   \\
 \text{(1, 2, (3, 4))} & \text{(1, 5, (2, 3))} & \text{(1, 3, (4, 5))} & \text{(2, 3, (4, 5))} &   \\
 \text{(1, 3, (2, 4))} & \text{(1, 5, (2, 3))} & \text{(1, 3, (4, 5))} & \text{(2, 3, (4, 5))} &   \\
 \text{(1, 4, (2, 3))} & \text{(1, 5, (2, 3))} & \text{(1, 3, (4, 5))} & \text{(2, 3, (4, 5))} & \text{(1, (2, 3), (4, 5))} \\
 \text{(1, 2, (3, 4))} & \text{(1, 2, (3, 5))} & \text{(1, 4, (3, 5))} & \text{(2, 3, (4, 5))} &   \\
 \text{(1, 3, (2, 4))} & \text{(1, 2, (3, 5))} & \text{(1, 4, (3, 5))} & \text{(2, 3, (4, 5))} &   \\
 \text{(1, 4, (2, 3))} & \text{(1, 2, (3, 5))} & \text{(1, 4, (3, 5))} & \text{(2, 3, (4, 5))} &   \\
 \text{(1, 2, (3, 4))} & \text{(1, 3, (2, 5))} & \text{(1, 4, (3, 5))} & \text{(2, 3, (4, 5))} &   \\
 \text{(1, 3, (2, 4))} & \text{(1, 3, (2, 5))} & \text{(1, 4, (3, 5))} & \text{(2, 3, (4, 5))} &   \\
 \text{(1, 4, (2, 3))} & \text{(1, 3, (2, 5))} & \text{(1, 4, (3, 5))} & \text{(2, 3, (4, 5))} &   \\
 \text{(1, 2, (3, 4))} & \text{(1, 5, (2, 3))} & \text{(1, 4, (3, 5))} & \text{(2, 3, (4, 5))} &   \\
 \text{(1, 3, (2, 4))} & \text{(1, 5, (2, 3))} & \text{(1, 4, (3, 5))} & \text{(2, 3, (4, 5))} &   \\
 \text{(1, 4, (2, 3))} & \text{(1, 5, (2, 3))} & \text{(1, 4, (3, 5))} & \text{(2, 3, (4, 5))} & \text{(1, ((2, 3), 5), 4)} \\
 \text{(1, 2, (3, 4))} & \text{(1, 2, (3, 5))} & \text{(1, 5, (3, 4))} & \text{(2, 3, (4, 5))} &   \\
 \text{(1, 3, (2, 4))} & \text{(1, 2, (3, 5))} & \text{(1, 5, (3, 4))} & \text{(2, 3, (4, 5))} &   \\
 \text{(1, 4, (2, 3))} & \text{(1, 2, (3, 5))} & \text{(1, 5, (3, 4))} & \text{(2, 3, (4, 5))} &   \\
 \text{(1, 2, (3, 4))} & \text{(1, 3, (2, 5))} & \text{(1, 5, (3, 4))} & \text{(2, 3, (4, 5))} &   \\
 \text{(1, 3, (2, 4))} & \text{(1, 3, (2, 5))} & \text{(1, 5, (3, 4))} & \text{(2, 3, (4, 5))} &   \\
 \text{(1, 4, (2, 3))} & \text{(1, 3, (2, 5))} & \text{(1, 5, (3, 4))} & \text{(2, 3, (4, 5))} &   \\
 \text{(1, 2, (3, 4))} & \text{(1, 5, (2, 3))} & \text{(1, 5, (3, 4))} & \text{(2, 3, (4, 5))} &   \\
 \text{(1, 3, (2, 4))} & \text{(1, 5, (2, 3))} & \text{(1, 5, (3, 4))} & \text{(2, 3, (4, 5))} &   \\
 \text{(1, 4, (2, 3))} & \text{(1, 5, (2, 3))} & \text{(1, 5, (3, 4))} & \text{(2, 3, (4, 5))} & \text{(1, ((2, 3), 4), 5)} \\
 \text{(1, 2, (3, 4))} & \text{(1, 2, (3, 5))} & \text{(1, 3, (4, 5))} & \text{(2, 4, (3, 5))} &   \\
 \text{(1, 3, (2, 4))} & \text{(1, 2, (3, 5))} & \text{(1, 3, (4, 5))} & \text{(2, 4, (3, 5))} &   \\
 \text{(1, 4, (2, 3))} & \text{(1, 2, (3, 5))} & \text{(1, 3, (4, 5))} & \text{(2, 4, (3, 5))} &   \\
 \text{(1, 2, (3, 4))} & \text{(1, 3, (2, 5))} & \text{(1, 3, (4, 5))} & \text{(2, 4, (3, 5))} &   \\
 \text{(1, 3, (2, 4))} & \text{(1, 3, (2, 5))} & \text{(1, 3, (4, 5))} & \text{(2, 4, (3, 5))} & \text{(1, ((2, 4), 5), 3)} \\
 \text{(1, 4, (2, 3))} & \text{(1, 3, (2, 5))} & \text{(1, 3, (4, 5))} & \text{(2, 4, (3, 5))} &   \\
 \text{(1, 2, (3, 4))} & \text{(1, 5, (2, 3))} & \text{(1, 3, (4, 5))} & \text{(2, 4, (3, 5))} &   \\
 \text{(1, 3, (2, 4))} & \text{(1, 5, (2, 3))} & \text{(1, 3, (4, 5))} & \text{(2, 4, (3, 5))} &   \\
 \text{(1, 4, (2, 3))} & \text{(1, 5, (2, 3))} & \text{(1, 3, (4, 5))} & \text{(2, 4, (3, 5))} &   \\
 \text{(1, 2, (3, 4))} & \text{(1, 2, (3, 5))} & \text{(1, 4, (3, 5))} & \text{(2, 4, (3, 5))} & \text{(1, 2, ((3, 5), 4))} \\
 \text{(1, 3, (2, 4))} & \text{(1, 2, (3, 5))} & \text{(1, 4, (3, 5))} & \text{(2, 4, (3, 5))} & \text{(1, (2, 4), (3, 5))} \\
 \text{(1, 4, (2, 3))} & \text{(1, 2, (3, 5))} & \text{(1, 4, (3, 5))} & \text{(2, 4, (3, 5))} & \text{(1, (2, (3, 5)), 4)} \\
 \text{(1, 2, (3, 4))} & \text{(1, 3, (2, 5))} & \text{(1, 4, (3, 5))} & \text{(2, 4, (3, 5))} &   \\
 \text{(1, 3, (2, 4))} & \text{(1, 3, (2, 5))} & \text{(1, 4, (3, 5))} & \text{(2, 4, (3, 5))} &   \\
 \text{(1, 4, (2, 3))} & \text{(1, 3, (2, 5))} & \text{(1, 4, (3, 5))} & \text{(2, 4, (3, 5))} &   \\
 \text{(1, 2, (3, 4))} & \text{(1, 5, (2, 3))} & \text{(1, 4, (3, 5))} & \text{(2, 4, (3, 5))} &   \\
 \text{(1, 3, (2, 4))} & \text{(1, 5, (2, 3))} & \text{(1, 4, (3, 5))} & \text{(2, 4, (3, 5))} &   \\
 \text{(1, 4, (2, 3))} & \text{(1, 5, (2, 3))} & \text{(1, 4, (3, 5))} & \text{(2, 4, (3, 5))} &   \\
 \text{(1, 2, (3, 4))} & \text{(1, 2, (3, 5))} & \text{(1, 5, (3, 4))} & \text{(2, 4, (3, 5))} &   \\
 \text{(1, 3, (2, 4))} & \text{(1, 2, (3, 5))} & \text{(1, 5, (3, 4))} & \text{(2, 4, (3, 5))} &   \\
 \text{(1, 4, (2, 3))} & \text{(1, 2, (3, 5))} & \text{(1, 5, (3, 4))} & \text{(2, 4, (3, 5))} &   \\
 \text{(1, 2, (3, 4))} & \text{(1, 3, (2, 5))} & \text{(1, 5, (3, 4))} & \text{(2, 4, (3, 5))} &   \\
 \text{(1, 3, (2, 4))} & \text{(1, 3, (2, 5))} & \text{(1, 5, (3, 4))} & \text{(2, 4, (3, 5))} &   \\
 \text{(1, 4, (2, 3))} & \text{(1, 3, (2, 5))} & \text{(1, 5, (3, 4))} & \text{(2, 4, (3, 5))} &   \\
 \text{(1, 2, (3, 4))} & \text{(1, 5, (2, 3))} & \text{(1, 5, (3, 4))} & \text{(2, 4, (3, 5))} &   \\
 \text{(1, 3, (2, 4))} & \text{(1, 5, (2, 3))} & \text{(1, 5, (3, 4))} & \text{(2, 4, (3, 5))} & \text{(1, ((2, 4), 3), 5)} \\
 \text{(1, 4, (2, 3))} & \text{(1, 5, (2, 3))} & \text{(1, 5, (3, 4))} & \text{(2, 4, (3, 5))} &   \\
 \text{(1, 2, (3, 4))} & \text{(1, 2, (3, 5))} & \text{(1, 3, (4, 5))} & \text{(2, 5, (3, 4))} &   \\
 \text{(1, 3, (2, 4))} & \text{(1, 2, (3, 5))} & \text{(1, 3, (4, 5))} & \text{(2, 5, (3, 4))} &   \\
 \text{(1, 4, (2, 3))} & \text{(1, 2, (3, 5))} & \text{(1, 3, (4, 5))} & \text{(2, 5, (3, 4))} &   \\
 \text{(1, 2, (3, 4))} & \text{(1, 3, (2, 5))} & \text{(1, 3, (4, 5))} & \text{(2, 5, (3, 4))} &   \\
 \text{(1, 3, (2, 4))} & \text{(1, 3, (2, 5))} & \text{(1, 3, (4, 5))} & \text{(2, 5, (3, 4))} & \text{(1, ((2, 5), 4), 3)} \\
 \text{(1, 4, (2, 3))} & \text{(1, 3, (2, 5))} & \text{(1, 3, (4, 5))} & \text{(2, 5, (3, 4))} &   \\
 \text{(1, 2, (3, 4))} & \text{(1, 5, (2, 3))} & \text{(1, 3, (4, 5))} & \text{(2, 5, (3, 4))} &   \\
 \text{(1, 3, (2, 4))} & \text{(1, 5, (2, 3))} & \text{(1, 3, (4, 5))} & \text{(2, 5, (3, 4))} &   \\
 \text{(1, 4, (2, 3))} & \text{(1, 5, (2, 3))} & \text{(1, 3, (4, 5))} & \text{(2, 5, (3, 4))} &   \\
 \text{(1, 2, (3, 4))} & \text{(1, 2, (3, 5))} & \text{(1, 4, (3, 5))} & \text{(2, 5, (3, 4))} &   \\
 \text{(1, 3, (2, 4))} & \text{(1, 2, (3, 5))} & \text{(1, 4, (3, 5))} & \text{(2, 5, (3, 4))} &   \\
 \text{(1, 4, (2, 3))} & \text{(1, 2, (3, 5))} & \text{(1, 4, (3, 5))} & \text{(2, 5, (3, 4))} &   \\
 \text{(1, 2, (3, 4))} & \text{(1, 3, (2, 5))} & \text{(1, 4, (3, 5))} & \text{(2, 5, (3, 4))} &   \\
 \text{(1, 3, (2, 4))} & \text{(1, 3, (2, 5))} & \text{(1, 4, (3, 5))} & \text{(2, 5, (3, 4))} &   \\
 \text{(1, 4, (2, 3))} & \text{(1, 3, (2, 5))} & \text{(1, 4, (3, 5))} & \text{(2, 5, (3, 4))} & \text{(1, ((2, 5), 3), 4)} \\
 \text{(1, 2, (3, 4))} & \text{(1, 5, (2, 3))} & \text{(1, 4, (3, 5))} & \text{(2, 5, (3, 4))} &   \\
 \text{(1, 3, (2, 4))} & \text{(1, 5, (2, 3))} & \text{(1, 4, (3, 5))} & \text{(2, 5, (3, 4))} &   \\
 \text{(1, 4, (2, 3))} & \text{(1, 5, (2, 3))} & \text{(1, 4, (3, 5))} & \text{(2, 5, (3, 4))} &   \\
 \text{(1, 2, (3, 4))} & \text{(1, 2, (3, 5))} & \text{(1, 5, (3, 4))} & \text{(2, 5, (3, 4))} & \text{(1, 2, ((3, 4), 5))} \\
 \text{(1, 3, (2, 4))} & \text{(1, 2, (3, 5))} & \text{(1, 5, (3, 4))} & \text{(2, 5, (3, 4))} &   \\
 \text{(1, 4, (2, 3))} & \text{(1, 2, (3, 5))} & \text{(1, 5, (3, 4))} & \text{(2, 5, (3, 4))} &   \\
 \text{(1, 2, (3, 4))} & \text{(1, 3, (2, 5))} & \text{(1, 5, (3, 4))} & \text{(2, 5, (3, 4))} & \text{(1, (2, 5), (3, 4))} \\
 \text{(1, 3, (2, 4))} & \text{(1, 3, (2, 5))} & \text{(1, 5, (3, 4))} & \text{(2, 5, (3, 4))} &   \\
 \text{(1, 4, (2, 3))} & \text{(1, 3, (2, 5))} & \text{(1, 5, (3, 4))} & \text{(2, 5, (3, 4))} &   \\
 \text{(1, 2, (3, 4))} & \text{(1, 5, (2, 3))} & \text{(1, 5, (3, 4))} & \text{(2, 5, (3, 4))} & \text{(1, (2, (3, 4)), 5)} \\
 \text{(1, 3, (2, 4))} & \text{(1, 5, (2, 3))} & \text{(1, 5, (3, 4))} & \text{(2, 5, (3, 4))} &   \\
 \text{(1, 4, (2, 3))} & \text{(1, 5, (2, 3))} & \text{(1, 5, (3, 4))} & \text{(2, 5, (3, 4))} &   \\
\end{tabular}
$
\caption{\label{tab:examplealltrees} All 81 combinations of input trees (given in the well-known \emph{Newick format} of nested parentheses, in which each pair of parentheses stands for an inner node of the tree and the elements in the parentheses denote its neighbors) for the collection of taxon sets $\mathcal{S}=\{ \{1,2,3,4\},\{1,2,3,5\},\{1,3,4,5\},\{2,3,4,5\}\}$ and all their respective possible supertrees. It can be seen that for each compatible choice of input trees, the supertree is unique. Whenever there exists no supertree (i.e., whenever the last column has no entry), the choice of input trees is not compatible.} 
\end{table}

{\small
\begin{algorithm}
\caption{Fixing taxon $c$-traceability: \enquote{graph-free} approach}\label{alg_generalFTT_WOGRAPH}
\LinesNumbered
 \SetKwInOut{Input}{Input}\SetKwInOut{Output}{Output}
 \vspace{0.15 cm}
 \Input{ $c \in \mathbb{N}_{\geq 1}$\\
 $n \in \mathbb{N}_{\geq c}$ \\
 set $\sS=\{Y_1,\ldots,Y_k\}: \ Y_i \subseteq X=\{1,\ldots,n\} \wedge \lvert Y_i\rvert \geq c 
 \ \forall i=1,\ldots,k$}
 \Output{\textsf{TRUE} if $\sS$ is fixing taxon $c$-traceable, \textsf{FALSE} else}
 \Init{}{
 $X \gets \{1,\ldots,n\}$\\
 \For{$i\gets 1$ \KwTo $\binom{n}{c}$}{
 $Q(i)\gets \mbox{$i^{th}$ $c$-tuple of $\binom{X}{c}$}$ \\
 $white(Q(i)) \gets 0$ }  
 $whiteCounter\gets 0$\\
 $newWhites \gets \emptyset$\\
}

\For{$i\gets 1$ \KwTo $ \binom{n}{c} $ }
{\For{$j \gets 1$ \KwTo $k$}{
\If{$Q(i) \subseteq Y_j$}{$white(Q(i))\gets 1$ \\ $newWhites \gets newWhites \cup \{Q(i)\}$\\
$whiteCounter \gets whiteCounter +1$\\
break\\ }
}}

\While{$whiteCounter < \binom{n}{c}$ \& $newWhites \neq \emptyset$}{

$tuple \gets newWhites(1)$\\
$X'\gets X \setminus tuple$ \\

\For{$i \gets 1$ \KwTo $|X'|=n-c$}{
$x \gets X'(i)$

$S \gets \binom{tuple\cup\{x\}}{c}$ \\ 
\tcc*[f]{set $S$ contains all possible $c$-tuples using only numbers from $tuple \cup \{x\}$ }\\
\tcc*[f]{Next: Check if there are $c$ white $c$-tuples in $S$}\\
\If { $\sum\limits_{j: Q(j) \in S} 
 white(Q(j))==c$}{ \For{$j\gets 1$ \KwTo $|S|$ }{\If{ $white(Q(j))==0$}{$white(Q(j))=1$\\ $newWhites \gets newWhites \cup \{Q(j)\}$\\
 $whiteCounter \gets whiteCounter+1$\\ break}}
 }
}
$newWhites\gets newWhites\setminus\{newWhites(1)\}$\\
}
\If{$whiteCounter==\binom{n}{c}$}{\Return{$\mathsf{TRUE}$}}
\Else{\Return{$\mathsf{FALSE}$}}
\end{algorithm}
}

\begin{table}\tiny
$ 
\begin{tabular}{cc@{\hspace{-0.001cm}}c@{\hspace{-0.001cm}}c@{\hspace{-0.001cm}}c@{\hspace{-0.001cm}}c@{\hspace{-0.001cm}}c@{\hspace{-0.001cm}}c@{\hspace{-0.001cm}}c@{\hspace{-0.001cm}}c@{\hspace{-0.001cm}}c@{\hspace{-0.001cm}}c}
&\{1,2,3,5\} & \{1,2,3,6\} & \{1,2,4,5\} & \{1,2,4,6\} &\{1,2,5,6\} &\{1,3,4,6\} &\{1,3,5,6\} &\{1,4,5,6\} &\{2,3,4,5\} &\{2,3,4,6\} &\{2,3,5,6\} \\
 \{\{1\},\{2\},\{3\},\{4,5,6\}\} & 1 & 1 & 0 & 0 & 0 & 0 & 0 & 0 & 0 & 0 & 0 \\
 \{\{1\},\{2\},\{3,4\},\{5,6\}\} & 1 & 1 & 1 & 1 & 0 & 0 & 0 & 0 & 0 & 0 & 0 \\
 \{\{1\},\{2\},\{3,5,6\},\{4\}\} & 0 & 0 & 1 & 1 & 0 & 0 & 0 & 0 & 0 & 0 & 0 \\
 \{\{1\},\{2\},\{3,4,5\},\{6\}\} & 0 & 1 & 0 & 1 & 1 & 0 & 0 & 0 & 0 & 0 & 0 \\
 \{\{1\},\{2\},\{3,6\},\{4,5\}\} & 1 & 0 & 0 & 1 & 1 & 0 & 0 & 0 & 0 & 0 & 0 \\
 \{\{1\},\{2\},\{3,4,6\},\{5\}\} & 1 & 0 & 1 & 0 & 1 & 0 & 0 & 0 & 0 & 0 & 0 \\
 \{\{1\},\{2\},\{3,5\},\{4,6\}\} & 0 & 1 & 1 & 0 & 1 & 0 & 0 & 0 & 0 & 0 & 0 \\
 \{\{1\},\{2,3\},\{4\},\{5,6\}\} & 0 & 0 & 1 & 1 & 0 & 1 & 0 & 0 & 0 & 0 & 0 \\
 \{\{1\},\{2,4\},\{3\},\{5,6\}\} & 1 & 1 & 0 & 0 & 0 & 1 & 0 & 0 & 0 & 0 & 0 \\
 \{\{1\},\{2,5,6\},\{3\},\{4\}\} & 0 & 0 & 0 & 0 & 0 & 1 & 0 & 0 & 0 & 0 & 0 \\
 \{\{1\},\{2,3\},\{4,5\},\{6\}\} & 0 & 0 & 0 & 1 & 1 & 1 & 1 & 0 & 0 & 0 & 0 \\
 \{\{1\},\{2,4,5\},\{3\},\{6\}\} & 0 & 1 & 0 & 0 & 0 & 1 & 1 & 0 & 0 & 0 & 0 \\
 \{\{1\},\{2,6\},\{3\},\{4,5\}\} & 1 & 0 & 0 & 0 & 0 & 1 & 1 & 0 & 0 & 0 & 0 \\
 \{\{1\},\{2,3\},\{4,6\},\{5\}\} & 0 & 0 & 1 & 0 & 1 & 0 & 1 & 0 & 0 & 0 & 0 \\
 \{\{1\},\{2,4,6\},\{3\},\{5\}\} & 1 & 0 & 0 & 0 & 0 & 0 & 1 & 0 & 0 & 0 & 0 \\
 \{\{1\},\{2,5\},\{3\},\{4,6\}\} & 0 & 1 & 0 & 0 & 0 & 0 & 1 & 0 & 0 & 0 & 0 \\
 \{\{1\},\{2,3,4\},\{5\},\{6\}\} & 0 & 0 & 0 & 0 & 1 & 0 & 1 & 1 & 0 & 0 & 0 \\
 \{\{1\},\{2,5\},\{3,4\},\{6\}\} & 0 & 1 & 0 & 1 & 0 & 0 & 1 & 1 & 0 & 0 & 0 \\
 \{\{1\},\{2,6\},\{3,4\},\{5\}\} & 1 & 0 & 1 & 0 & 0 & 0 & 1 & 1 & 0 & 0 & 0 \\
 \{\{1\},\{2,3,5\},\{4\},\{6\}\} & 0 & 0 & 0 & 1 & 0 & 1 & 0 & 1 & 0 & 0 & 0 \\
 \{\{1\},\{2,4\},\{3,5\},\{6\}\} & 0 & 1 & 0 & 0 & 1 & 1 & 0 & 1 & 0 & 0 & 0 \\
 \{\{1\},\{2,6\},\{3,5\},\{4\}\} & 0 & 0 & 1 & 0 & 0 & 1 & 0 & 1 & 0 & 0 & 0 \\
 \{\{1\},\{2,3,6\},\{4\},\{5\}\} & 0 & 0 & 1 & 0 & 0 & 0 & 0 & 1 & 0 & 0 & 0 \\
 \{\{1\},\{2,4\},\{3,6\},\{5\}\} & 1 & 0 & 0 & 0 & 1 & 0 & 0 & 1 & 0 & 0 & 0 \\
 \{\{1\},\{2,5\},\{3,6\},\{4\}\} & 0 & 0 & 0 & 1 & 0 & 0 & 0 & 1 & 0 & 0 & 0 \\
 \{\{1,2\},\{3\},\{4\},\{5,6\}\} & 0 & 0 & 0 & 0 & 0 & 1 & 0 & 0 & 1 & 1 & 0 \\
 \{\{1,3\},\{2\},\{4\},\{5,6\}\} & 0 & 0 & 1 & 1 & 0 & 0 & 0 & 0 & 1 & 1 & 0 \\
 \{\{1,4\},\{2\},\{3\},\{5,6\}\} & 1 & 1 & 0 & 0 & 0 & 0 & 0 & 0 & 1 & 1 & 0 \\
 \{\{1,5,6\},\{2\},\{3\},\{4\}\} & 0 & 0 & 0 & 0 & 0 & 0 & 0 & 0 & 1 & 1 & 0 \\
 \{\{1,2\},\{3\},\{4,5\},\{6\}\} & 0 & 0 & 0 & 0 & 0 & 1 & 1 & 0 & 0 & 1 & 1 \\
 \{\{1,3\},\{2\},\{4,5\},\{6\}\} & 0 & 0 & 0 & 1 & 1 & 0 & 0 & 0 & 0 & 1 & 1 \\
 \{\{1,4,5\},\{2\},\{3\},\{6\}\} & 0 & 1 & 0 & 0 & 0 & 0 & 0 & 0 & 0 & 1 & 1 \\
 \{\{1,6\},\{2\},\{3\},\{4,5\}\} & 1 & 0 & 0 & 0 & 0 & 0 & 0 & 0 & 0 & 1 & 1 \\
 \{\{1,2\},\{3\},\{4,6\},\{5\}\} & 0 & 0 & 0 & 0 & 0 & 0 & 1 & 0 & 1 & 0 & 1 \\
 \{\{1,3\},\{2\},\{4,6\},\{5\}\} & 0 & 0 & 1 & 0 & 1 & 0 & 0 & 0 & 1 & 0 & 1 \\
 \{\{1,4,6\},\{2\},\{3\},\{5\}\} & 1 & 0 & 0 & 0 & 0 & 0 & 0 & 0 & 1 & 0 & 1 \\
 \{\{1,5\},\{2\},\{3\},\{4,6\}\} & 0 & 1 & 0 & 0 & 0 & 0 & 0 & 0 & 1 & 0 & 1 \\
 \{\{1,2\},\{3,4\},\{5\},\{6\}\} & 0 & 0 & 0 & 0 & 0 & 0 & 1 & 1 & 0 & 0 & 1 \\
 \{\{1,3,4\},\{2\},\{5\},\{6\}\} & 0 & 0 & 0 & 0 & 1 & 0 & 0 & 0 & 0 & 0 & 1 \\
 \{\{1,5\},\{2\},\{3,4\},\{6\}\} & 0 & 1 & 0 & 1 & 0 & 0 & 0 & 0 & 0 & 0 & 1 \\
 \{\{1,6\},\{2\},\{3,4\},\{5\}\} & 1 & 0 & 1 & 0 & 0 & 0 & 0 & 0 & 0 & 0 & 1 \\
 \{\{1,2\},\{3,5\},\{4\},\{6\}\} & 0 & 0 & 0 & 0 & 0 & 1 & 0 & 1 & 0 & 1 & 0 \\
 \{\{1,3,5\},\{2\},\{4\},\{6\}\} & 0 & 0 & 0 & 1 & 0 & 0 & 0 & 0 & 0 & 1 & 0 \\
 \{\{1,4\},\{2\},\{3,5\},\{6\}\} & 0 & 1 & 0 & 0 & 1 & 0 & 0 & 0 & 0 & 1 & 0 \\
 \{\{1,6\},\{2\},\{3,5\},\{4\}\} & 0 & 0 & 1 & 0 & 0 & 0 & 0 & 0 & 0 & 1 & 0 \\
 \{\{1,2\},\{3,6\},\{4\},\{5\}\} & 0 & 0 & 0 & 0 & 0 & 0 & 0 & 1 & 1 & 0 & 0 \\
 \{\{1,3,6\},\{2\},\{4\},\{5\}\} & 0 & 0 & 1 & 0 & 0 & 0 & 0 & 0 & 1 & 0 & 0 \\
 \{\{1,4\},\{2\},\{3,6\},\{5\}\} & 1 & 0 & 0 & 0 & 1 & 0 & 0 & 0 & 1 & 0 & 0 \\
 \{\{1,5\},\{2\},\{3,6\},\{4\}\} & 0 & 0 & 0 & 1 & 0 & 0 & 0 & 0 & 1 & 0 & 0 \\
 \{\{1,2,3\},\{4\},\{5\},\{6\}\} & 0 & 0 & 0 & 0 & 0 & 0 & 0 & 1 & 0 & 0 & 0 \\
 \{\{1,4\},\{2,3\},\{5\},\{6\}\} & 0 & 0 & 0 & 0 & 1 & 0 & 1 & 0 & 0 & 0 & 0 \\
 \{\{1,5\},\{2,3\},\{4\},\{6\}\} & 0 & 0 & 0 & 1 & 0 & 1 & 0 & 0 & 0 & 0 & 0 \\
 \{\{1,6\},\{2,3\},\{4\},\{5\}\} & 0 & 0 & 1 & 0 & 0 & 0 & 0 & 0 & 0 & 0 & 0 \\
 \{\{1,2,4\},\{3\},\{5\},\{6\}\} & 0 & 0 & 0 & 0 & 0 & 0 & 1 & 0 & 0 & 0 & 1 \\
 \{\{1,3\},\{2,4\},\{5\},\{6\}\} & 0 & 0 & 0 & 0 & 1 & 0 & 0 & 1 & 0 & 0 & 1 \\
 \{\{1,5\},\{2,4\},\{3\},\{6\}\} & 0 & 1 & 0 & 0 & 0 & 1 & 0 & 0 & 0 & 0 & 1 \\
 \{\{1,6\},\{2,4\},\{3\},\{5\}\} & 1 & 0 & 0 & 0 & 0 & 0 & 0 & 0 & 0 & 0 & 1 \\
 \{\{1,2,5\},\{3\},\{4\},\{6\}\} & 0 & 0 & 0 & 0 & 0 & 1 & 0 & 0 & 0 & 1 & 0 \\
 \{\{1,3\},\{2,5\},\{4\},\{6\}\} & 0 & 0 & 0 & 1 & 0 & 0 & 0 & 1 & 0 & 1 & 0 \\
 \{\{1,4\},\{2,5\},\{3\},\{6\}\} & 0 & 1 & 0 & 0 & 0 & 0 & 1 & 0 & 0 & 1 & 0 \\
 \{\{1,6\},\{2,5\},\{3\},\{4\}\} & 0 & 0 & 0 & 0 & 0 & 0 & 0 & 0 & 0 & 1 & 0 \\
 \{\{1,2,6\},\{3\},\{4\},\{5\}\} & 0 & 0 & 0 & 0 & 0 & 0 & 0 & 0 & 1 & 0 & 0 \\
 \{\{1,3\},\{2,6\},\{4\},\{5\}\} & 0 & 0 & 1 & 0 & 0 & 0 & 0 & 1 & 1 & 0 & 0 \\
 \{\{1,4\},\{2,6\},\{3\},\{5\}\} & 1 & 0 & 0 & 0 & 0 & 0 & 1 & 0 & 1 & 0 & 0 \\
 \{\{1,5\},\{2,6\},\{3\},\{4\}\} & 0 & 0 & 0 & 0 & 0 & 1 & 0 & 0 & 1 & 0 & 0 \\
\end{tabular}
$\caption{\label{tab_MainEx}Overview of all 65 partitions and the quadruple set $\sS= \{1,2,3,5\},\{1,2,3,6\},\{1,2,4,5\},\{1,2,4,6\},\{1,2,5,6\},\{1,3,4,6\},\{1,3,5,6\},\{1,4,5,6\},\{2,3,4,5\},$ $\{2,3,4,6\},\{2,3,5,6\}$. Whenever a quadruple of this set covers a partition, the respective entry is marked with a 1, otherwise with a 0. This table shows that $\sS $ is phylogenetically decisive.} 
\end{table}

 \begin{table}\tiny
$ 
\begin{tabular}{cc@{\hspace{-0.001cm}}c@{\hspace{-0.001cm}}c@{\hspace{-0.001cm}}c@{\hspace{-0.001cm}}c@{\hspace{-0.001cm}}c@{\hspace{-0.001cm}}c@{\hspace{-0.001cm}}c@{\hspace{-0.001cm}}c}
 & \{1,2,3,5\} & \{1,2,4,5\} & \{1,2,4,6\} & \{1,3,4,6\} & \{1,3,5,6\} & \{1,4,5,6\}
   & \{2,3,4,5\} & \{2,3,4,6\} & \{2,3,5,6\} \\
 \{\{1\},\{2\},\{3\},\{4,5,6\}\} & 1 & 0 & 0 & 0 & 0 & 0 & 0 & 0 & 0 \\
 \{\{1\},\{2\},\{3,4\},\{5,6\}\} & 1 & 1 & 1 & 0 & 0 & 0 & 0 & 0 & 0 \\
 \{\{1\},\{2\},\{3,5,6\},\{4\}\} & 0 & 1 & 1 & 0 & 0 & 0 & 0 & 0 & 0 \\
 \{\{1\},\{2\},\{3,4,5\},\{6\}\} & 0 & 0 & 1 & 0 & 0 & 0 & 0 & 0 & 0 \\
 \{\{1\},\{2\},\{3,6\},\{4,5\}\} & 1 & 0 & 1 & 0 & 0 & 0 & 0 & 0 & 0 \\
 \{\{1\},\{2\},\{3,4,6\},\{5\}\} & 1 & 1 & 0 & 0 & 0 & 0 & 0 & 0 & 0 \\
 \{\{1\},\{2\},\{3,5\},\{4,6\}\} & 0 & 1 & 0 & 0 & 0 & 0 & 0 & 0 & 0 \\
 \{\{1\},\{2,3\},\{4\},\{5,6\}\} & 0 & 1 & 1 & 1 & 0 & 0 & 0 & 0 & 0 \\
 \{\{1\},\{2,4\},\{3\},\{5,6\}\} & 1 & 0 & 0 & 1 & 0 & 0 & 0 & 0 & 0 \\
 \{\{1\},\{2,5,6\},\{3\},\{4\}\} & 0 & 0 & 0 & 1 & 0 & 0 & 0 & 0 & 0 \\
 \{\{1\},\{2,3\},\{4,5\},\{6\}\} & 0 & 0 & 1 & 1 & 1 & 0 & 0 & 0 & 0 \\
 \{\{1\},\{2,4,5\},\{3\},\{6\}\} & 0 & 0 & 0 & 1 & 1 & 0 & 0 & 0 & 0 \\
 \{\{1\},\{2,6\},\{3\},\{4,5\}\} & 1 & 0 & 0 & 1 & 1 & 0 & 0 & 0 & 0 \\
 \{\{1\},\{2,3\},\{4,6\},\{5\}\} & 0 & 1 & 0 & 0 & 1 & 0 & 0 & 0 & 0 \\
 \{\{1\},\{2,4,6\},\{3\},\{5\}\} & 1 & 0 & 0 & 0 & 1 & 0 & 0 & 0 & 0 \\
 \{\{1\},\{2,5\},\{3\},\{4,6\}\} & 0 & 0 & 0 & 0 & 1 & 0 & 0 & 0 & 0 \\
 \{\{1\},\{2,3,4\},\{5\},\{6\}\} & 0 & 0 & 0 & 0 & 1 & 1 & 0 & 0 & 0 \\
 \{\{1\},\{2,5\},\{3,4\},\{6\}\} & 0 & 0 & 1 & 0 & 1 & 1 & 0 & 0 & 0 \\
 \{\{1\},\{2,6\},\{3,4\},\{5\}\} & 1 & 1 & 0 & 0 & 1 & 1 & 0 & 0 & 0 \\
 \{\{1\},\{2,3,5\},\{4\},\{6\}\} & 0 & 0 & 1 & 1 & 0 & 1 & 0 & 0 & 0 \\
 \{\{1\},\{2,4\},\{3,5\},\{6\}\} & 0 & 0 & 0 & 1 & 0 & 1 & 0 & 0 & 0 \\
 \{\{1\},\{2,6\},\{3,5\},\{4\}\} & 0 & 1 & 0 & 1 & 0 & 1 & 0 & 0 & 0 \\
 \{\{1\},\{2,3,6\},\{4\},\{5\}\} & 0 & 1 & 0 & 0 & 0 & 1 & 0 & 0 & 0 \\
 \{\{1\},\{2,4\},\{3,6\},\{5\}\} & 1 & 0 & 0 & 0 & 0 & 1 & 0 & 0 & 0 \\
 \{\{1\},\{2,5\},\{3,6\},\{4\}\} & 0 & 0 & 1 & 0 & 0 & 1 & 0 & 0 & 0 \\
 \{\{1,2\},\{3\},\{4\},\{5,6\}\} & 0 & 0 & 0 & 1 & 0 & 0 & 1 & 1 & 0 \\
 \{\{1,3\},\{2\},\{4\},\{5,6\}\} & 0 & 1 & 1 & 0 & 0 & 0 & 1 & 1 & 0 \\
 \{\{1,4\},\{2\},\{3\},\{5,6\}\} & 1 & 0 & 0 & 0 & 0 & 0 & 1 & 1 & 0 \\
 \{\{1,5,6\},\{2\},\{3\},\{4\}\} & 0 & 0 & 0 & 0 & 0 & 0 & 1 & 1 & 0 \\
 \{\{1,2\},\{3\},\{4,5\},\{6\}\} & 0 & 0 & 0 & 1 & 1 & 0 & 0 & 1 & 1 \\
 \{\{1,3\},\{2\},\{4,5\},\{6\}\} & 0 & 0 & 1 & 0 & 0 & 0 & 0 & 1 & 1 \\
 \{\{1,4,5\},\{2\},\{3\},\{6\}\} & 0 & 0 & 0 & 0 & 0 & 0 & 0 & 1 & 1 \\
 \{\{1,6\},\{2\},\{3\},\{4,5\}\} & 1 & 0 & 0 & 0 & 0 & 0 & 0 & 1 & 1 \\
 \{\{1,2\},\{3\},\{4,6\},\{5\}\} & 0 & 0 & 0 & 0 & 1 & 0 & 1 & 0 & 1 \\
 \{\{1,3\},\{2\},\{4,6\},\{5\}\} & 0 & 1 & 0 & 0 & 0 & 0 & 1 & 0 & 1 \\
 \{\{1,4,6\},\{2\},\{3\},\{5\}\} & 1 & 0 & 0 & 0 & 0 & 0 & 1 & 0 & 1 \\
 \{\{1,5\},\{2\},\{3\},\{4,6\}\} & 0 & 0 & 0 & 0 & 0 & 0 & 1 & 0 & 1 \\
 \{\{1,2\},\{3,4\},\{5\},\{6\}\} & 0 & 0 & 0 & 0 & 1 & 1 & 0 & 0 & 1 \\
 \{\{1,3,4\},\{2\},\{5\},\{6\}\} & 0 & 0 & 0 & 0 & 0 & 0 & 0 & 0 & 1 \\
 \{\{1,5\},\{2\},\{3,4\},\{6\}\} & 0 & 0 & 1 & 0 & 0 & 0 & 0 & 0 & 1 \\
 \{\{1,6\},\{2\},\{3,4\},\{5\}\} & 1 & 1 & 0 & 0 & 0 & 0 & 0 & 0 & 1 \\
 \{\{1,2\},\{3,5\},\{4\},\{6\}\} & 0 & 0 & 0 & 1 & 0 & 1 & 0 & 1 & 0 \\
 \{\{1,3,5\},\{2\},\{4\},\{6\}\} & 0 & 0 & 1 & 0 & 0 & 0 & 0 & 1 & 0 \\
 \{\{1,4\},\{2\},\{3,5\},\{6\}\} & 0 & 0 & 0 & 0 & 0 & 0 & 0 & 1 & 0 \\
 \{\{1,6\},\{2\},\{3,5\},\{4\}\} & 0 & 1 & 0 & 0 & 0 & 0 & 0 & 1 & 0 \\
 \{\{1,2\},\{3,6\},\{4\},\{5\}\} & 0 & 0 & 0 & 0 & 0 & 1 & 1 & 0 & 0 \\
 \{\{1,3,6\},\{2\},\{4\},\{5\}\} & 0 & 1 & 0 & 0 & 0 & 0 & 1 & 0 & 0 \\
 \{\{1,4\},\{2\},\{3,6\},\{5\}\} & 1 & 0 & 0 & 0 & 0 & 0 & 1 & 0 & 0 \\
 \{\{1,5\},\{2\},\{3,6\},\{4\}\} & 0 & 0 & 1 & 0 & 0 & 0 & 1 & 0 & 0 \\
 \{\{1,2,3\},\{4\},\{5\},\{6\}\} & 0 & 0 & 0 & 0 & 0 & 1 & 0 & 0 & 0 \\
 \{\{1,4\},\{2,3\},\{5\},\{6\}\} & 0 & 0 & 0 & 0 & 1 & 0 & 0 & 0 & 0 \\
 \{\{1,5\},\{2,3\},\{4\},\{6\}\} & 0 & 0 & 1 & 1 & 0 & 0 & 0 & 0 & 0 \\
 \{\{1,6\},\{2,3\},\{4\},\{5\}\} & 0 & 1 & 0 & 0 & 0 & 0 & 0 & 0 & 0 \\
 \{\{1,2,4\},\{3\},\{5\},\{6\}\} & 0 & 0 & 0 & 0 & 1 & 0 & 0 & 0 & 1 \\
 \{\{1,3\},\{2,4\},\{5\},\{6\}\} & 0 & 0 & 0 & 0 & 0 & 1 & 0 & 0 & 1 \\
 \{\{1,5\},\{2,4\},\{3\},\{6\}\} & 0 & 0 & 0 & 1 & 0 & 0 & 0 & 0 & 1 \\
 \{\{1,6\},\{2,4\},\{3\},\{5\}\} & 1 & 0 & 0 & 0 & 0 & 0 & 0 & 0 & 1 \\
 \{\{1,2,5\},\{3\},\{4\},\{6\}\} & 0 & 0 & 0 & 1 & 0 & 0 & 0 & 1 & 0 \\
 \{\{1,3\},\{2,5\},\{4\},\{6\}\} & 0 & 0 & 1 & 0 & 0 & 1 & 0 & 1 & 0 \\
 \{\{1,4\},\{2,5\},\{3\},\{6\}\} & 0 & 0 & 0 & 0 & 1 & 0 & 0 & 1 & 0 \\
 \{\{1,6\},\{2,5\},\{3\},\{4\}\} & 0 & 0 & 0 & 0 & 0 & 0 & 0 & 1 & 0 \\
 \{\{1,2,6\},\{3\},\{4\},\{5\}\} & 0 & 0 & 0 & 0 & 0 & 0 & 1 & 0 & 0 \\
 \{\{1,3\},\{2,6\},\{4\},\{5\}\} & 0 & 1 & 0 & 0 & 0 & 1 & 1 & 0 & 0 \\
 \{\{1,4\},\{2,6\},\{3\},\{5\}\} & 1 & 0 & 0 & 0 & 1 & 0 & 1 & 0 & 0 \\
 \{\{1,5\},\{2,6\},\{3\},\{4\}\} & 0 & 0 & 0 & 1 & 0 & 0 & 1 & 0 & 0 \\
\end{tabular}
$\caption{\label{tab_jannesEx}Overview of all 65 partitions and the quadruple set $\sS= \{1,2,3,5\},\{1,2,4,5\},\{1,2,4,6\},\{1,3,4,6\},\{1,3,5,6\},\{1,4,5,6\},\{2,3,4,5\},\{2,3,4,6\},\{2,3,5,6\}$. Whenever a quadruple of this set covers a partition, the respective entry is marked with a 1, otherwise with a 0. This table shows that $\sS $ is phylogenetically decisive.}
\end{table}

\begin{table}\tiny
$ 
\begin{tabular}{cc@{\hspace{-0.001cm}}c@{\hspace{-0.001cm}}c@{\hspace{-0.001cm}}c@{\hspace{-0.001cm}}c@{\hspace{-0.001cm}}c@{\hspace{-0.001cm}}c@{\hspace{-0.001cm}}c@{\hspace{-0.001cm}}c@{\hspace{-0.001cm}}c}
  & \{1,2,3,4\} & \{1,2,3,5\} & \{1,2,3,6\} & \{1,2,4,5\} & \{1,2,4,6\} & \{1,3,4,5\}
   & \{1,3,5,6\} & \{2,3,4,5\} & \{2,4,5,6\} & \{3,4,5,6\} \\
\{\{1\},\{2\},\{3\},\{4,5,6\}\} & 1 & 1 & 1 & 0 & 0 & 0 & 0 & 0 & 0 & 0 \\
 \{\{1\},\{2\},\{3,4\},\{5,6\}\} & 0 & 1 & 1 & 1 & 1 & 0 & 0 & 0 & 0 & 0 \\
 \{\{1\},\{2\},\{3,5,6\},\{4\}\} & 1 & 0 & 0 & 1 & 1 & 0 & 0 & 0 & 0 & 0 \\
 \{\{1\},\{2\},\{3,4,5\},\{6\}\} & 0 & 0 & 1 & 0 & 1 & 0 & 0 & 0 & 0 & 0 \\
 \{\{1\},\{2\},\{3,6\},\{4,5\}\} & 1 & 1 & 0 & 0 & 1 & 0 & 0 & 0 & 0 & 0 \\
 \{\{1\},\{2\},\{3,4,6\},\{5\}\} & 0 & 1 & 0 & 1 & 0 & 0 & 0 & 0 & 0 & 0 \\
 \{\{1\},\{2\},\{3,5\},\{4,6\}\} & 1 & 0 & 1 & 1 & 0 & 0 & 0 & 0 & 0 & 0 \\
 \{\{1\},\{2,3\},\{4\},\{5,6\}\} & 0 & 0 & 0 & 1 & 1 & 1 & 0 & 0 & 0 & 0 \\
 \{\{1\},\{2,4\},\{3\},\{5,6\}\} & 0 & 1 & 1 & 0 & 0 & 1 & 0 & 0 & 0 & 0 \\
 \{\{1\},\{2,5,6\},\{3\},\{4\}\} & 1 & 0 & 0 & 0 & 0 & 1 & 0 & 0 & 0 & 0 \\
 \{\{1\},\{2,3\},\{4,5\},\{6\}\} & 0 & 0 & 0 & 0 & 1 & 0 & 1 & 0 & 0 & 0 \\
 \{\{1\},\{2,4,5\},\{3\},\{6\}\} & 0 & 0 & 1 & 0 & 0 & 0 & 1 & 0 & 0 & 0 \\
 \{\{1\},\{2,6\},\{3\},\{4,5\}\} & 1 & 1 & 0 & 0 & 0 & 0 & 1 & 0 & 0 & 0 \\
 \{\{1\},\{2,3\},\{4,6\},\{5\}\} & 0 & 0 & 0 & 1 & 0 & 1 & 1 & 0 & 0 & 0 \\
 \{\{1\},\{2,4,6\},\{3\},\{5\}\} & 0 & 1 & 0 & 0 & 0 & 1 & 1 & 0 & 0 & 0 \\
 \{\{1\},\{2,5\},\{3\},\{4,6\}\} & 1 & 0 & 1 & 0 & 0 & 1 & 1 & 0 & 0 & 0 \\
 \{\{1\},\{2,3,4\},\{5\},\{6\}\} & 0 & 0 & 0 & 0 & 0 & 0 & 1 & 0 & 0 & 0 \\
 \{\{1\},\{2,5\},\{3,4\},\{6\}\} & 0 & 0 & 1 & 0 & 1 & 0 & 1 & 0 & 0 & 0 \\
 \{\{1\},\{2,6\},\{3,4\},\{5\}\} & 0 & 1 & 0 & 1 & 0 & 0 & 1 & 0 & 0 & 0 \\
 \{\{1\},\{2,3,5\},\{4\},\{6\}\} & 0 & 0 & 0 & 0 & 1 & 0 & 0 & 0 & 0 & 0 \\
 \{\{1\},\{2,4\},\{3,5\},\{6\}\} & 0 & 0 & 1 & 0 & 0 & 0 & 0 & 0 & 0 & 0 \\
 \{\{1\},\{2,6\},\{3,5\},\{4\}\} & 1 & 0 & 0 & 1 & 0 & 0 & 0 & 0 & 0 & 0 \\
 \{\{1\},\{2,3,6\},\{4\},\{5\}\} & 0 & 0 & 0 & 1 & 0 & 1 & 0 & 0 & 0 & 0 \\
 \{\{1\},\{2,4\},\{3,6\},\{5\}\} & 0 & 1 & 0 & 0 & 0 & 1 & 0 & 0 & 0 & 0 \\
 \{\{1\},\{2,5\},\{3,6\},\{4\}\} & 1 & 0 & 0 & 0 & 1 & 1 & 0 & 0 & 0 & 0 \\
 \{\{1,2\},\{3\},\{4\},\{5,6\}\} & 0 & 0 & 0 & 0 & 0 & 1 & 0 & 1 & 0 & 0 \\
 \{\{1,3\},\{2\},\{4\},\{5,6\}\} & 0 & 0 & 0 & 1 & 1 & 0 & 0 & 1 & 0 & 0 \\
 \{\{1,4\},\{2\},\{3\},\{5,6\}\} & 0 & 1 & 1 & 0 & 0 & 0 & 0 & 1 & 0 & 0 \\
 \{\{1,5,6\},\{2\},\{3\},\{4\}\} & 1 & 0 & 0 & 0 & 0 & 0 & 0 & 1 & 0 & 0 \\
 \{\{1,2\},\{3\},\{4,5\},\{6\}\} & 0 & 0 & 0 & 0 & 0 & 0 & 1 & 0 & 0 & 0 \\
 \{\{1,3\},\{2\},\{4,5\},\{6\}\} & 0 & 0 & 0 & 0 & 1 & 0 & 0 & 0 & 0 & 0 \\
 \{\{1,4,5\},\{2\},\{3\},\{6\}\} & 0 & 0 & 1 & 0 & 0 & 0 & 0 & 0 & 0 & 0 \\
 \{\{1,6\},\{2\},\{3\},\{4,5\}\} & 1 & 1 & 0 & 0 & 0 & 0 & 0 & 0 & 0 & 0 \\
 \{\{1,2\},\{3\},\{4,6\},\{5\}\} & 0 & 0 & 0 & 0 & 0 & 1 & 1 & 1 & 0 & 0 \\
 \{\{1,3\},\{2\},\{4,6\},\{5\}\} & 0 & 0 & 0 & 1 & 0 & 0 & 0 & 1 & 0 & 0 \\
 \{\{1,4,6\},\{2\},\{3\},\{5\}\} & 0 & 1 & 0 & 0 & 0 & 0 & 0 & 1 & 0 & 0 \\
 \{\{1,5\},\{2\},\{3\},\{4,6\}\} & 1 & 0 & 1 & 0 & 0 & 0 & 0 & 1 & 0 & 0 \\
 \{\{1,2\},\{3,4\},\{5\},\{6\}\} & 0 & 0 & 0 & 0 & 0 & 0 & 1 & 0 & 1 & 0 \\
 \{\{1,3,4\},\{2\},\{5\},\{6\}\} & 0 & 0 & 0 & 0 & 0 & 0 & 0 & 0 & 1 & 0 \\
 \{\{1,5\},\{2\},\{3,4\},\{6\}\} & 0 & 0 & 1 & 0 & 1 & 0 & 0 & 0 & 1 & 0 \\
 \{\{1,6\},\{2\},\{3,4\},\{5\}\} & 0 & 1 & 0 & 1 & 0 & 0 & 0 & 0 & 1 & 0 \\
 \{\{1,2\},\{3,5\},\{4\},\{6\}\} & 0 & 0 & 0 & 0 & 0 & 0 & 0 & 0 & 1 & 0 \\
 \{\{1,3,5\},\{2\},\{4\},\{6\}\} & 0 & 0 & 0 & 0 & 1 & 0 & 0 & 0 & 1 & 0 \\
 \{\{1,4\},\{2\},\{3,5\},\{6\}\} & 0 & 0 & 1 & 0 & 0 & 0 & 0 & 0 & 1 & 0 \\
 \{\{1,6\},\{2\},\{3,5\},\{4\}\} & 1 & 0 & 0 & 1 & 0 & 0 & 0 & 0 & 1 & 0 \\
 \{\{1,2\},\{3,6\},\{4\},\{5\}\} & 0 & 0 & 0 & 0 & 0 & 1 & 0 & 1 & 1 & 0 \\
 \{\{1,3,6\},\{2\},\{4\},\{5\}\} & 0 & 0 & 0 & 1 & 0 & 0 & 0 & 1 & 1 & 0 \\
 \{\{1,4\},\{2\},\{3,6\},\{5\}\} & 0 & 1 & 0 & 0 & 0 & 0 & 0 & 1 & 1 & 0 \\
 \{\{1,5\},\{2\},\{3,6\},\{4\}\} & 1 & 0 & 0 & 0 & 1 & 0 & 0 & 1 & 1 & 0 \\
 \{\{1,2,3\},\{4\},\{5\},\{6\}\} & 0 & 0 & 0 & 0 & 0 & 0 & 0 & 0 & 1 & 1 \\
 \{\{1,4\},\{2,3\},\{5\},\{6\}\} & 0 & 0 & 0 & 0 & 0 & 0 & 1 & 0 & 1 & 1 \\
 \{\{1,5\},\{2,3\},\{4\},\{6\}\} & 0 & 0 & 0 & 0 & 1 & 0 & 0 & 0 & 1 & 1 \\
 \{\{1,6\},\{2,3\},\{4\},\{5\}\} & 0 & 0 & 0 & 1 & 0 & 1 & 0 & 0 & 1 & 1 \\
 \{\{1,2,4\},\{3\},\{5\},\{6\}\} & 0 & 0 & 0 & 0 & 0 & 0 & 1 & 0 & 0 & 1 \\
 \{\{1,3\},\{2,4\},\{5\},\{6\}\} & 0 & 0 & 0 & 0 & 0 & 0 & 0 & 0 & 0 & 1 \\
 \{\{1,5\},\{2,4\},\{3\},\{6\}\} & 0 & 0 & 1 & 0 & 0 & 0 & 0 & 0 & 0 & 1 \\
 \{\{1,6\},\{2,4\},\{3\},\{5\}\} & 0 & 1 & 0 & 0 & 0 & 1 & 0 & 0 & 0 & 1 \\
 \{\{1,2,5\},\{3\},\{4\},\{6\}\} & 0 & 0 & 0 & 0 & 0 & 0 & 0 & 0 & 0 & 1 \\
 \{\{1,3\},\{2,5\},\{4\},\{6\}\} & 0 & 0 & 0 & 0 & 1 & 0 & 0 & 0 & 0 & 1 \\
 \{\{1,4\},\{2,5\},\{3\},\{6\}\} & 0 & 0 & 1 & 0 & 0 & 0 & 1 & 0 & 0 & 1 \\
 \{\{1,6\},\{2,5\},\{3\},\{4\}\} & 1 & 0 & 0 & 0 & 0 & 1 & 0 & 0 & 0 & 1 \\
 \{\{1,2,6\},\{3\},\{4\},\{5\}\} & 0 & 0 & 0 & 0 & 0 & 1 & 0 & 1 & 0 & 1 \\
 \{\{1,3\},\{2,6\},\{4\},\{5\}\} & 0 & 0 & 0 & 1 & 0 & 0 & 0 & 1 & 0 & 1 \\
 \{\{1,4\},\{2,6\},\{3\},\{5\}\} & 0 & 1 & 0 & 0 & 0 & 0 & 1 & 1 & 0 & 1 \\
 \{\{1,5\},\{2,6\},\{3\},\{4\}\} & 1 & 0 & 0 & 0 & 0 & 0 & 0 & 1 & 0 & 1 \\
\end{tabular}
$\caption{\label{tab_maxDecNonFTT}Overview of all 65 partitions and the quadruple set $\sS= \{\{1, 2, 3, 4\}, \{1, 2, 3, 5\}, \{1, 2, 3, 6\}, \{1, 2, 4, 5\}, \{1, 2, 4, 
  6\}, \{1, 3, 4, 5\}, \{1, 3, 5, 6\}, \{2, 3, 4, 5\}, \{2, 4, 5, 6\}, \{3, 4, 5, 6\}\}$. Whenever a quadruple of this set covers a partition, the respective entry is marked with a 1, otherwise with a 0. This table shows that $\sS $ is phylogenetically decisive.}
\end{table}

\end{document}